\documentclass[AMA,STIX1COL]{WileyNJD-v2}

\articletype{Article Type}%

\usepackage{amsmath} 
\usepackage{amssymb}  
\usepackage{graphicx}
\usepackage{amsmath}
\usepackage{epstopdf}
\usepackage{mathtools}
\usepackage{dsfont}
\usepackage{url}
\usepackage{color}
\usepackage{booktabs} 
\usepackage{multirow}   
\usepackage{balance}    
\usepackage{soul}   
\usepackage{hyperref}   
\newcommand{\tKN}{\mathcal{K}_{N}}         

\DeclareMathOperator{\Pre}{Pre}

\newcommand{\beq}{\begin{equation}}
\newcommand{\eeq}{\end{equation}}

\usepackage{bm}

\usepackage{hyperref}

\begin{document}

\title{Minimum Time  Learning Model Predictive Control }

\author[1]{Ugo Rosolia*}

\author[2]{Francesco Borrelli}

\authormark{Ugo Rosolia and Francesco Borrelli}

\address[1]{\orgdiv{Department of Mechanical and Civil Engineering}, \orgname{California Institute of Technology}, \orgaddress{\state{California}, \country{USA}}}

\address[2]{\orgdiv{Department of Mechanical Engineering}, \orgname{University of California Berkeley}, \orgaddress{\state{California}, \country{USA}}}

\corres{*Ugo Rosolia, B100 Gates-Thomas Laboratory
MC 104-44 Pasadena, CA 91125. \email{urosolia@caltech.edu}}


\abstract[Summary]{In this paper we present a Learning Model Predictive Control (LMPC) strategy for linear and nonlinear time optimal control problems. Our work builds on existing LMPC methodologies and it guarantees finite time convergence properties for the closed-loop system. We show how to construct a time varying safe set and terminal cost function using closed-loop data. The resulting LMPC policy is time varying and it guarantees recursive constraint satisfaction and non-decreasing performance. 
Computational efficiency is obtained by  convexifing the time varying safe set and time varying terminal cost function. We demonstrate that, for a class of nonlinear system and convex constraints, the convex LMPC formulation guarantees recursive constraint satisfaction and non-decreasing performance. 
Finally, we illustrate the effectiveness of the proposed strategies on minimum time obstacle avoidance and racing examples.}

\keywords{predictive control, learning model predictive control, minimum time, iterative improvement}


\maketitle

\section{Introduction}
In time optimal control problems, the goal of the controller is to steer the system from the starting point $x_S$ to the terminal point $x_F$ in minimum time, while satisfying state and input constraints.
These problems have been studied since the 1950s~\cite{bogner1954investigation, bellman1956bang, gamkrelidze1961theory, lasalle1959time} and it was shown that the optimal input strategy is a piece-wise function which saturates the input constraints~\cite{bogner1954investigation, bellman1956bang, gamkrelidze1961theory}. Furthermore, while investigating the solution to time optimal control problems, researches formalized the maximum principle which 
describes the first order necessary optimality conditions~\cite{boltyanskiy1961theory, liberzon2011calculus}.

For linear systems, time optimal control problems can be solved applying the maximum principle. However, for nonlinear systems the optimality conditions are hard to solve, as those are described by a two-point boundary value problem for a system of nonlinear differential equations~\cite{liberzon2011calculus}. For this reason, several approaches have been proposed to approximate the solution to time optimal control problems. These strategies can be divided in three different categories: $i)$ \textit{hierarchical approaches}, where in the first step a collision-free path is generated and afterwards it is computed the speed profile which minimizes the travel time along the path~\cite{bobrow1985time, kapania2016sequential, shiller1992computation,nagy2019sequential, rajan1985minimum, verscheure2009time},  $ii)$ \textit{maximum principle-based strategies}, which exploit the necessary optimality conditions~\cite{meier1990efficient,manor2018time,scott2017time} and $iii)$ \textit{iterative optimization strategies}, where 
the original control problem is approximated solving sequentially or in parallel simpler optimization problems~\cite{graham2015minimum, bobrow1988optimal, verschueren2016time, al2016hierarchical}. A comprehensive literature review is out of the scope of this work. 
In the following, we focus on iterative optimization strategies, because the proposed approach falls into this category. In \cite{graham2015minimum} the time optimal control problem is posed as a constrained nonlinear optimization problem and it is solved using a variable-order Legendre-Gauss-Radau
method, where the initial guesses for the algorithm are obtained by solving a sequence of modified optimal control problems. 
A different approach was proposed in~\cite{bobrow1988optimal} where the path is parametrized using basis function which are amenable for optimization.
The authors in \cite{verschueren2016time} first proposed a smooth spatial system reformulation for the autonomous racing time optimal control problem. 
Afterwards, they used a nonlinear optimization solver based on a SQP algorithm to compute an optimal solution.
In \cite{al2016hierarchical} the authors used at each time step a Model Predictive Controller (MPC) to compute a trajectory which drives the system from the current state to the end state.

We propose to iteratively approximate the optimal solution to time optimal problems.
In particular, we formulate the time optimal problem as an iterative control task, where at each iteration $j$ the goal of the controller is to steer the system from the starting point $x_S$ to the terminal point $x_F$ in minimum time. Several strategies have been proposed to iteratively synthesize MPC policies for iterative tasks~\cite{c34,c33,c35,c4,liu2013nonlinear}. However, these approaches assume that the goal of the controller is to track a given reference trajectory, which is not available in minimum time optimal control problems. The work in~\cite{tamar2017learning} presented a reference-free ILC strategy, where the stage cost of an MPC is learned after each iteration. The authors demonstrated the effectiveness of the control strategy on several navigation and manipulation examples, however the control strategy does not consider a terminal constraint set and terminal cost function which are required to guarantee safety and convergence to the goal set in minimum time problems.
Therefore, we build on the reference-free Learning Model Predictive Control (LMPC) strategy~\cite{rosolia2017learning}, where the terminal constraint set and terminal cost function are estimated from data.
In particular after each iteration $j$, the closed-loop data are stored and used to estimate $i)$ a \textit{safe set} of states from which the control task can be completed using a known policy $\pi^j$ and $ii)$ a \textit{value function} which approximates the closed-loop cost associated with the policy~$\pi^j$. These safe set and value function are used as a terminal constraint set and terminal cost function to synthesize the LMPC policy at the next iteration $j+1$. 

The first contribution of this work is to design a LMPC scheme for nonlinear systems where the safe set and the approximated value function are time varying. 
At each time $t$ of iteration $j$, the proposed time varying LMPC uses a subset of the stored data to compute the control action and it allows us to reduce the computational burden associated with time invariant LMPC methodologies~\cite{rosolia2017learning}, as demonstrated in the result section.
We show that the proposed strategy guarantees safety, finite time convergence and non-decreasing performance with respect to previous executions of the control task. We assume that the model is known, and therefore the learning process of the time varying safe set and approximated value function can be performed in simulations. Furthermore, we show that, when the system dynamics are subject to bounded additive disturbances, the proposed approach can be combined with standard MPC strategies to design robust LMPC policies.
When a model is not available, system identification strategies may be used to learn an uncertain model~\cite{aswani2013provably,berkenkamp2016safe,koller2018learning,hewing2018cautious,kocijan2004gaussian,rosolia2019learning, terzi2018learning,hewing2019learning, ostafew2014learning, berkenkamp2015safe}.
In this case, the safe set and value function are constructed using experimental data and their properties are probabilistic, as discussed in~\cite{rosolia2019sample}. The second contribution of this work is to propose a relaxed LMPC formulation which is based on a convexified time varying safe set and cost function. This strategy enables the reduction of the computational burden while guaranteeing safety and non-decreasing performance for a specific class of nonlinear system and convex constraints. Furthermore, we show that the same properties hold for nonlinear systems, if a sufficient condition on the stored states and the system dynamics is satisfied. Compared to the relaxed formulations from~\cite{rosolia2017learninglinear} and \cite{CDCRepetitiveRacing}, the proposed time varying strategy is applicable to a broader class of dynamical systems and it is tailored to minimum time problems.
Finally, we test the proposed strategies on nonlinear time optimal control problems. We show that the proposed LMPC is able to match the performance of the strategy from~\cite{rosolia2017learning}, while being computationally faster. 


\section{Problem Formulation}\label{sec:probFormulation}
Consider the nonlinear system
\begin{equation}\label{eq:system}
x_{t+1}^j=f(x_t^j,u_t^j),
\end{equation}
where at time $t$ of iteration $j$ the state $x_t^j \in \mathbb{R}^n$ and the input $u_t^j \in \mathbb{R}^d$. Furthermore, the system is subject to the following state and input constraints
\begin{equation}\label{eq:stateInputConstr}
\begin{aligned}
x_t^j \in \mathcal{X} \text{ and } u_t^j\in \mathcal{U}, \forall t \geq 0, \forall j \geq 0.
\end{aligned}
\end{equation}

The goal of the controller is to solve the following minimum time optimal control problem
\begin{equation}\label{eq:minTimeOPC}
    \begin{aligned}
        \min_{{T^j}, u_0^j, \ldots, u_{{T^j}-1}^j } & \quad \sum_{t = 0}^{{T^j}-1} 1 \\
        \text{s.t. } \quad~ & \quad x_{t+1}^j = f(x_t^j, u_t^j), \forall t = [0,\dots, {T^j}-1] \\
        & \quad x_t^j \in \mathcal{X},\ u_t^j\in \mathcal{U}, \forall t = [0,\dots, {T^j}-1] \\
        & \quad x^j_{{T^j}} = x_F,\\
        & \quad x_0^j = x_S
    \end{aligned}
\end{equation}
where the goal state $x_F$ is an unforced equilibrium point for system~\eqref{eq:system}, i.e., $f(x_F, 0) = x_F$.

In this paper we propose to solve Problem~\eqref{eq:minTimeOPC} iteratively. In particular, at each iteration we drive the system from the starting point $x_S$ to the terminal state $x_F$ and we store the closed-loop trajectories. After completion of the $j$th iteration, these trajectories are used to synthesize a control policy 
for the next iteration $j+1$. 
We show that the proposed iterative design strategy guarantees recursive constraint satisfaction and non-decreasing closed-loop performance.
Next, we define the safe set and value function approximation which will be used in the controller design.

\begin{remark}
In the following, we focus on iterative tasks where the initial condition is the same at each iteration, i.e., $x_0^j=x_S,~\forall j\geq 0$. Afterwards, in Section~\ref{sec:perturbed} we discuss the properties of the proposed control strategy when the initial condition is perturbed at each iteration. Finally, in Section~\ref{sec:dubinsRacing} we test the controller for different initial conditions.
\end{remark}

\section{Safe Set and Value Function Approximation}\label{sec:SSandVfun}

At each $j$th iteration of the control task, we store the closed-loop trajectories and the associated input sequences. In particular, at the $j$th iteration we define the vectors
\begin{equation}\label{eq:storedTrajectories}
\begin{aligned}
 {\bf{u}}^j ~ &= ~ [u_0^j,\ldots,~u_{ T^j}^j], \\
 {\bf{x}}^j ~ &= ~ [x_0^j,\ldots,~x_{T^j}^j], 
\end{aligned}
\end{equation}
where $x_t^j$ and $u_t^j$ are the state and input of system~\eqref{eq:system}. In \eqref{eq:storedTrajectories}, $T^j$ denotes the time at which the closed-loop system reached the terminal state, i.e., $x_{T^j} = x_F$.

\subsection{Time Varying Safe Set}
We use the stored data to build time varying safe sets, which will be used in the controller design to guarantee recursive constraint satisfaction. 
First, we define the time varying safe set at iteration $j$ as
\begin{equation}\label{eq:SS}
    \mathcal{SS}^{j}_t = \bigcup_{i = 0}^j \bigcup_{k = \delta^{j,i}_t}^{T^i} x_{k}^i,
\end{equation}
where, for $T^{j,*} = \min_{k \in \{ 0, \ldots, j \}}  T^k $,
\begin{equation}\label{eq:deltaTime}
    {\delta^{j,i}_t = \min(t + T^i - T^{j,*},T^i).}
\end{equation} 
Definition~\eqref{eq:deltaTime} implies that if at time $t$ of the $j$th iteration $x_t^j = x_{\delta_t^{j,i}}^i \neq x_{T^i}^i$, then system~\eqref{eq:system} can be steered along the $i$th trajectory to reach $x_{T^i}^i = x_F$ in $(T^{j,*}-t)$ time steps.
Therefore, at each time $t$ the time varying safe set $\mathcal{SS}^{j}_t$ collects the stored states from which system~\eqref{eq:system} can reach the terminal state $x_F$ in at most $(T^{j,*}-t)$ time steps.
A representation of the time varying safe set for a two-dimensional system is shown in Figure~\ref{fig:SSillustration}. 
We notice that, by definition, if a state $x_t^i$ belongs to $\mathcal{SS}^{j}_t$, then there exists a feasible control action $u_t^i \in \mathcal{U}$ which keeps the evolution of the nonlinear system~\eqref{eq:system} within the time varying safe set at the next time step $t+1$, i.e., $f(x_t^i, u_t^i) \in \mathcal{SS}^{j}_{t+1}$. This property will be used in the controller design to guarantee that state and input constraints~\eqref{eq:stateInputConstr} are recursively satisfied.

\begin{figure}[h!]
	\centering \includegraphics[width=0.5\columnwidth]{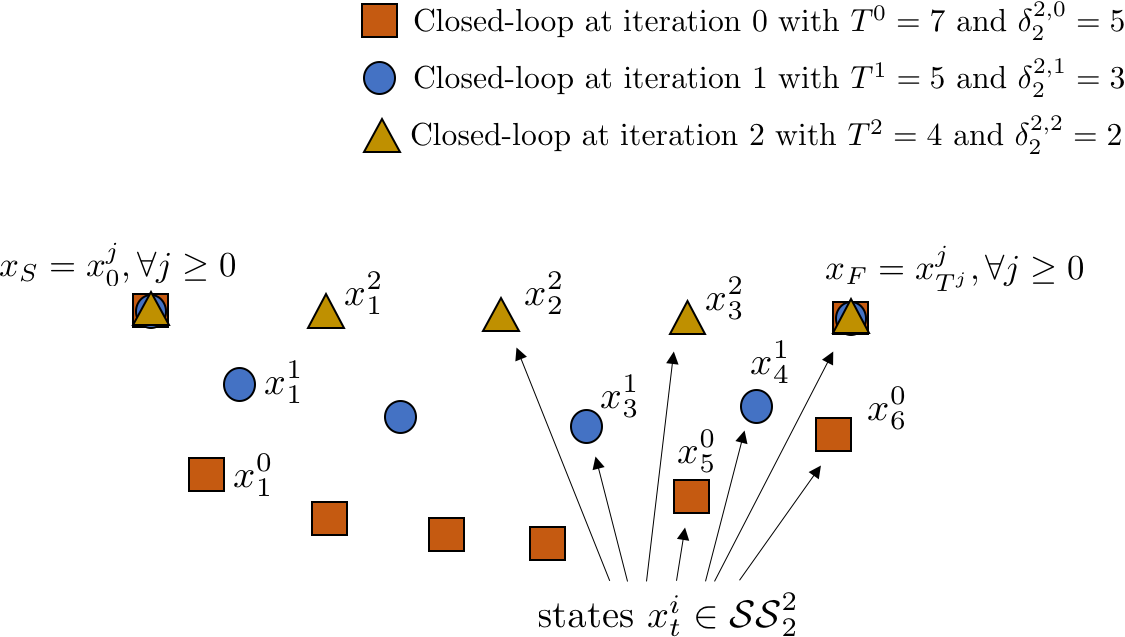}
	\caption{Representation of the time varying safe set $\mathcal{SS}^{2}_2$. We notice that just a subset of the stored states are used to define $\mathcal{SS}^{2}_2$. Furthermore, we notice that from all states $x_t^i \in \mathcal{SS}_2^2$ system~\eqref{eq:system} can be steered to $x_F$ in at most $T^{j,*}-t=2$ time steps.}\label{fig:SSillustration}
\end{figure}

Finally, at each time $t$ we define the local convex safe set as the convex hull of $\mathcal{SS}^{j}_t$ from~\eqref{eq:SS},
\begin{equation}\label{eq:CS}
\begin{aligned}
\mathcal{CS}^{j}_t &= \text{Conv}\big( \mathcal{SS}^{j}_t \big) \\
&= \big\{x\in \mathbb{R}^n : \exists [ \lambda^0_{\delta_t^{j,0}},..., \lambda_{T^j}^j ] \geq 0,  \sum_{i = 0}^j \sum_{k = \delta^{j,i}_t}^{T^i} \lambda_{k}^i = 1, \sum_{i = 0}^j \sum_{k = \delta^{j,i}_t}^{T^i} \lambda_{k}^i x_k^i = x\big\}.
\end{aligned}
\end{equation}
Later on we will show that for a class of nonlinear systems,  if a state $x_t^i$ belongs to $\mathcal{CS}^{j}_t$, then there exists a feasible control action $u_t^i \in \mathcal{U}$ which keeps the evolution of the nonlinear system~\eqref{eq:system} within the convex safe set at the next time step $t+1$. For such class of nonlinear systems, $\mathcal{CS}^{j}_t$ can be used to synthesize controllers which guarantee state and input constraint satisfaction at all time instants.

\begin{remark}\label{remark:terminalSet}
When the goal of the controller is to reach an invariant set $\mathcal{X}_F$ in minimum time, it is still possible to use the proposed iterative control strategy. In this case one should replace $x_{T^i}^i = x_F$ with $\mathcal{X}_F$ in definition~\eqref{eq:SS}. 
\end{remark}

\subsection{Time Varying Value Function Approximation}
In this section, we show how to construct $Q$-functions which approximate the cost-to-go over the safe set and convex safe set. These functions will be used in the controller design to guarantee non-decreasing performance at each iteration.

We define the cost-to-go associated with the stored state $x_t^j$ from~\eqref{eq:storedTrajectories},
\begin{equation}\label{eq:costToGo}
\begin{aligned}
J_{t\rightarrow T^j}^j(x_t^j) = ~ \sum\limits_{k=t}^{T^j} \mathds{1}_{x_F}(x_k^j),
\end{aligned}
\end{equation}
where the indicator function 
\begin{equation*}
    \mathds{1}_{x_F}(x) = \begin{cases}1 & \mbox{If } x_F \neq x \\
    0 & \mbox{Else}\end{cases}.
\end{equation*}
The above cost-to-go represents the time steps needed to steer system~\eqref{eq:system} from $x_t^j$ to the terminal state $x_F$ along the $j$th trajectory, and it is used to construct the function $Q^{j}_t(\cdot)$, defined over the safe set $\mathcal{SS}^{j}_t$,
\begin{equation}\label{eq:Qfunction}
\begin{aligned}
Q^{j}_t(x) = \min\limits_{\substack{i \in \{0, \ldots, j\} \\ k \in \{\delta_t^{j,i}, \ldots, T^i\} } } & \quad J^i_{k\rightarrow T^i}(x_k^i)\\
\text{s.t. } \quad~ & \quad x = x_k^i \in \mathcal{SS}^{j}_t.
\end{aligned}
\end{equation}
The function $Q^{j}_t(\cdot)$ assigns to every point in the safe set $\mathcal{SS}^{j}_t$ from~\eqref{eq:SS} the minimum cost-to-go along the stored trajectories from~\eqref{eq:storedTrajectories}, i.e.,
\begin{equation*}\label{eq:SStraj}
\begin{aligned}
\forall x \in \mathcal{SS}^{j}_t, Q^{j}_t(x) &= J^{(i^*)}_{k^*\rightarrow T^{(i^*)}}(x) = \sum_{k=k^*}^{T^{(i^*)}} \mathds{1}_{x_F}\big(x_k^{(i^*)}\big),
\end{aligned}
\end{equation*}
where ${i^*}$ and ${k^*}$ are the minimizers in~\eqref{eq:Qfunction}:
\begin{equation}\label{eq:argMin}
\begin{aligned}
[{i^*}, {k^*}] =\mathop{\mathrm{argmin}}\limits_{ \substack{i \in \{0, \ldots, j\} \\ k \in \{\delta_t^{j,i}, \ldots, T^i\} } } & \quad J^i_{k\rightarrow T^i}(x_k^i)\\
\text{s.t.} \quad ~~ & \quad x = x_k^i \in \mathcal{SS}^{j}_t.
\end{aligned}
\end{equation}

Finally, we define the convex $Q$-function over the convex safe set $\mathcal{CS}^{j}_t$ from~\eqref{eq:CS},
\begin{equation}\label{eq:cvxQfunction}
\begin{aligned}
\bar{Q}^{j}_t(x) = \min\limits_{ [\lambda_{\delta_t^{j,0}}^0, \ldots, \lambda_{T^j}^j] \geq 0 } & \quad \sum_{i = 0}^j \sum_{k = \delta_t^{j,i}}^{T^i} \lambda_{k}^i J^i_{k\rightarrow T^i}(x_{k}^{i})\\
\text{s.t. } \quad ~~ & \quad \sum_{i = 0}^j \sum_{k = \delta_t^{j,i}}^{T^i} \lambda_{k}^i x_k^i = x \\
& \quad \sum_{i = 0}^j \sum_{k = \delta_t^{j,i}}^{T^i} \lambda_{k}^i = 1,
\end{aligned}
\end{equation}
where $\delta^{j,i}_t$ is defined in~\eqref{eq:deltaTime}. 
The convex $Q$-function $\bar{Q}^{j}_t(\cdot)$ is simply a piecewise-affine interpolation of the $Q$-function from~\eqref{eq:Qfunction} over the convex safe set, as shown in Figure~\ref{fig:localQfunillustration}. 
In Section~\ref{sec:properties}, we will show that the convex $Q$-function can be used to guarantee non-decreasing performance for a particular class of nonlinear systems.

\begin{figure}[h!]
	\centering \includegraphics[width=0.5\columnwidth]{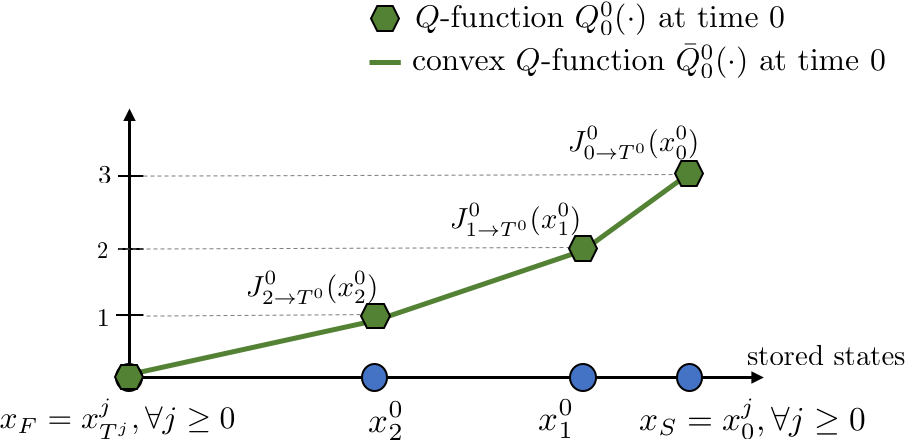}
	\caption{Representation of the $Q$-function $Q^{0}_0(\cdot)$ and convex $Q$-function $\bar Q^{0}_0(\cdot)$. We notice that the $Q$-function $Q^{0}_0(\cdot)$ is defined over a set of discrete data points, whereas the convex $Q$-function $\bar Q^{0}_0(\cdot)$ is defined over the convex safe set.}\label{fig:localQfunillustration}
\end{figure}


\section{Learning Model Predictive Control Design}
In this section, we describe the controller design. 
We propose a Learning Model Predictive Control (LMPC) strategy for nonlinear systems which guarantees recursive constraint satisfaction and non-decreasing performance at each iteration. 
Computing the control action from the LMPC policy is expensive. 
For this reason, we also present a relaxed LMPC policy, which allows us to reduce the computational cost and it guarantees recursive constraint satisfaction and non-decreasing performance for a class of nonlinear systems.

\subsection{LMPC: Mixed Integer Formulation}\label{sec:LMPC_MIP}

At each time $t$ of the $j$th iteration, we solve the following finite time optimal control problem,
\begin{equation}\label{eq:FTOCP}
	\begin{aligned}
		J_{t\rightarrow t+N}^{\scalebox{0.4}{LMPC},j}(x_t^j) = \min_{{\bf{U}}_t^j } \quad & \bigg[  \sum_{k=t}^{t+N-1}  \mathds{1}_{x_F}(x_{k|t}^j) +Q_{t+N}^{j-1}(x_{t+N|t}^j)\bigg] \\
		\text{s.t.}\quad 
		&x_{k+1|t}^j=f( x_{k|t}^j, u_{k|t}^j ), \forall k = t, \cdots, t+N-1 \\
		&x_{k|t}^j \in \mathcal{X}, u_{k|t}^j \in \mathcal{U}, \forall k = t, \cdots, t+N-1\\ 
		& x_{t+N|t}^j \in ~\mathcal{SS}_{t+N}^{j-1}\\
		& x_{t|t}^j=x_t^j\\
	\end{aligned}
\end{equation}
where ${\bf{U}}_t^j = [u_{t|t}^j,\ldots,u_{t+N-1|t}^j] \in \mathbb{R}^{d \times N}$. The solution to the above finite time optimal control problem steers system~\eqref{eq:system} from $x_t^j$ to the time varying safe set $\mathcal{SS}^{j-1}_{t+N}$ while satisfying state, input and dynamic constraints. Let 
\begin{equation}\label{eq:optSolution}
\begin{aligned}
    {\bf{U}}_t^{j,*} &= [u_{t|t}^{j,*},\ldots,u_{t+N-1|t}^{j,*}]
\end{aligned}
\end{equation}
be the optimal solution to~\eqref{eq:FTOCP} at time $t$ of the $j$th iteration. Then, we apply to system \eqref{eq:system} the first element of the optimizer vector,
\begin{equation}\label{eq:policyLMPC}
    u_t^j = \pi^{\scalebox{0.4}{LMPC},j}_t(x_t^j)  = u_{t|t}^{j,*}.
\end{equation}
The finite time optimal control problem \eqref{eq:FTOCP} is repeated at time $t+1$, based on the new state $x_{t+1|t+1} = x_{t+1}^j$, until the iteration is terminated when $x_{t+1}^j = x_F$.

Computing the control action from the LMPC policy~\eqref{eq:policyLMPC} requires to solve a mixed-integer optimization problem, as $\mathcal{SS}^{j-1}_t$ is a set of discrete states. In particular, the number of integer variables grows as more iterations are stored. In Section~\ref{sec:dataReduction}, we will show that the computational cost may be reduced synthesizing the LMPC policy~\eqref{eq:policyLMPC} using a subset of the stored data. Finally, in the result section we will show that the number of data points used in the synthesis process affects the performance improvement at each iteration. Therefore, there is a trade-off between the online computational burden and the number of iterations needed to reach desirable closed-loop performance. 

\subsection{Relaxed LMPC: Nonlinear Formulation}\label{sec:LMPC_Nonlinear}
In this section, we present a relaxed LMPC constructed using the convex safe set $\mathcal{CS}^{j-1}_t$ from~\eqref{eq:CS}.
At each time $t$ of the $j$th iteration, we solve the following finite time optimal control problem
\begin{equation}\label{eq:RelaxedFTOCP}
	\begin{aligned}
			\bar J_{t\rightarrow t+N}^{\scalebox{0.4}{LMPC},j}(x_t^j) =\min_{{\bf{U}}_t^j, \boldsymbol{\lambda}_t^j \geq 0 } & \quad \sum_{k=t}^{t+N-1}  \mathds{1}_{x_F}(x_{k|t}^j) + \sum_{i = 0}^{j-1} \sum_{k = \delta^{j-1,i}_{t+N}}^{T^i}  \lambda_{k}^i J^i_{k\rightarrow T^i}(x_{k}^{i})  \\
			\text{s.t.}\quad
			&\quad x_{k+1|t}^j=f( x_{k|t}^j, u_{k|t}^j ), \forall k = t, \cdots, t+N-1 \\
			&\quad x_{k|t}^j \in \mathcal{X}, u_{k|t}^j \in \mathcal{U}, \forall k = t, \cdots, t+N-1\\ 
			&\quad \sum_{i = 0}^{j-1} \sum_{k = \delta^{j-1,i}_{t+N}}^{T^i} \lambda_{k}^i x_k^i = x^{j}_{t+N|t},\\
			&\quad \sum_{i = 0}^{j-1} \sum_{k = \delta^{j-1,i}_{t+N}}^{T^i} \lambda_{k}^i = 1 \\
			&\quad  x_{t|t}^j=x_t^j\\
	\end{aligned}
\end{equation}
where ${\bf{U}}_t^j = [u_{t|t}^j,\ldots,u_{t+N-1|t}^j] \in \mathbb{R}^{d \times N}$ and the vector $\boldsymbol{\lambda}_t^j = [\lambda_{0}^0, \ldots, \lambda_{T^{j-1}}^{j-1}] \in \mathbb{R}^{\Pi_{i=0}^{j-1} T^i}$ parameterizes the terminal constraint set $\mathcal{CS}^{j-1}_{t+N}$ and terminal cost $\bar Q^{j-1}_{t+N}(\cdot)$. Let 
\begin{equation}\label{eq:optSolutionRelaxed}
\begin{aligned}
{\bf{U}}_t^{j,*} &= [u_{t|t}^{j,*},\ldots,u_{t+N-1|t}^{j,*}] \\
\boldsymbol{\lambda}_t^{j,*} &= [\lambda_{0}^{j,*}, \ldots, \lambda_{T^{h-1}}^{j,*}].
\end{aligned}
\end{equation}
be the optimal solution to~\eqref{eq:RelaxedFTOCP} at time $t$ of the $j$th iteration. Then, we apply to the system \eqref{eq:system} the first element of the optimal input sequence,
\begin{equation}\label{eq:policyLMPCRelaxed}
u_t^j = \bar \pi^{\scalebox{0.4}{LMPC},j}_t(x_t^j)  = u_{t|t}^{j,*}.
\end{equation}

Notice that the terminal constraint $x_{t+N|t}^j \in \mathcal{CS}_t^j$ in~\eqref{eq:RelaxedFTOCP} is enforced using equality constraint on the state $x_{t+N|t}^j$ and multiplier $\lambda_k^i$, and inequality constraints on the multipliers $\lambda_k^i$. Therefore the computation burden is reduced with respect to the LMPC from Section~\ref{sec:LMPC_MIP}. In the next section, we will show that for a class of nonlinear systems the relaxed LMPC~\eqref{eq:RelaxedFTOCP} and \eqref{eq:policyLMPCRelaxed} guarantees safety and non-decreasing performance.

\section{Properties}\label{sec:properties}
This section describes the properties of the proposed control strategies. We show that the LMPC guarantees constraint satisfaction at all time instants, convergence in finite time to $x_F$ and non-decreasing performance. Furthermore, we demonstrate that the same properties are guaranteed when the relaxed LMPC is in closed-loop with a specific class of nonlinear systems or when a sufficient condition on the stored data and the system dynamics is satisfied.

\subsection{Recursive Feasibility}
We assume that a feasible trajectory which drives the system from the starting point $x_S$ to the terminal state $x_F$ is given. Afterwards, we show that the controller recursively satisfies state and input constraints~\eqref{eq:stateInputConstr}.

\begin{assumption}\label{ass:feasTrajGiven}
At iteration $j=0$, we are given the closed-loop trajectory and associated input sequence
\begin{equation*}
\begin{aligned}
 {\bf{x}}^0 = [x_0^0,\ldots,~x_{T^0}^0]\text{ and }{\bf{u}}^0 = [u_0^0,\ldots,~u_{ T^0}^0], 
\end{aligned}
\end{equation*}
which satisfy state and input constraints~\eqref{eq:stateInputConstr}. Furthermore, we have that $x_0^0 = x_S$ and $x_{T^0}^0 =x_F$.
\end{assumption}


\begin{theorem}\label{th:recFeas}
Consider system \eqref{eq:system} controlled by the \mbox{LMPC} \eqref{eq:FTOCP} and \eqref{eq:policyLMPC}.
Let $\mathcal{SS}^{j}_t$ be the time varying safe set at iteration $j$ as defined in \eqref{eq:SS}. Let Assumption~\ref{ass:feasTrajGiven} hold and assume that $x_0^j = x_S ~\forall j\geq 0$. Then at every iteration $j\geq1$ the LMPC~\eqref{eq:FTOCP} and \eqref{eq:policyLMPC} is feasible for all $t \geq 0$ when \eqref{eq:policyLMPC} is applied to system~\eqref{eq:system}.
\end{theorem}

\begin{proof}
The proof follows from standard MPC arguments. Assume that the LMPC \eqref{eq:FTOCP} and \eqref{eq:policyLMPC} is feasible at time $t$, let \eqref{eq:optSolution} be the optimal solution and $x_{t+N|t}^{j,*} = x_{k^*}^{i^*}$, where $x_{k^*}^{i^*}$ is defined in~\eqref{eq:argMin}.
Then, we have that the following state trajectory and associated input sequence
\begin{equation}\label{eq:candidateSolution}
\begin{aligned}
    &[x_{t+1|t}^{j,*}, \ldots, x_{t+N|t}^{j,*} = x_{k^*}^{i^*}, x_{k^* + 1}^{i^*}] \\
    &[u_{t+1|t}^{j,*}, \ldots, u_{t+N-1|t}^{j,*}, u_{k^*}^{i^*}],
\end{aligned}
\end{equation}
satisfy input and state constraints~\eqref{eq:stateInputConstr} and the LMPC at time $t+1$ of the $j$th iteration is feasible.\\
Now notice that at time $t=0$ of the $j$th iteration, the state trajectory and associated input sequence
\begin{equation}\label{eq:feasTrajTime0}
\begin{aligned}
    [x_0^{j-1}, \ldots, x_{N}^{j-1}] \text{ and } [u_0^{j-1}, \ldots, u_{N-1}^{j-1}],
\end{aligned}
\end{equation}
satisfy input and state constraints. Therefore the LMPC is feasible at time $t=0$ of the $j$th iteration. Finally, we conclude by induction that the LMPC \eqref{eq:FTOCP} and \eqref{eq:policyLMPC} is feasible for all $t \geq 0$ and iteration $j\geq1$.
\end{proof}

Next, we consider a specific class of nonlinear systems which satisfies the following assumption.
\begin{assumption}\label{ass:simplified}
Given any $P$ states $x^{(i)} \in \mathcal{X}$ and inputs $u^{(i)} \in \mathcal{U}$ for $i\in \{1,\ldots, P\}$, we have that $\forall x \in \text{Conv}(x^{(1)}, \ldots, x^{(P)})$ there exists $ u \in \mathcal{U}$ such that
\begin{equation*}
    f(x,u) \in \text{Conv}\big( f(x^{(1)},u^{(1)}), \ldots, f(x^{(P)},u^{(P)}) \big)
\end{equation*}
where $f(\cdot, \cdot)$ is defined in~\eqref{eq:system}.
\end{assumption}


Finally, we show that if Assumption~\ref{ass:simplified} is satisfied and the constraint sets in~\eqref{eq:stateInputConstr} are convex, then the relaxed LMPC~\eqref{eq:RelaxedFTOCP} and \eqref{eq:policyLMPCRelaxed} in closed-loop with system~\eqref{eq:system} guarantees recursive state and input constraint satisfaction.

\begin{assumption}\label{ass:convexity}
The state and  input constraint sets $\mathcal{X}$ and $\mathcal{U}$ in~\eqref{eq:stateInputConstr} are convex.
\end{assumption}

\begin{theorem}\label{th:recFeasSimple}
Consider system \eqref{eq:system} controlled by the relaxed \mbox{LMPC} \eqref{eq:RelaxedFTOCP} and \eqref{eq:policyLMPCRelaxed}.
Let $\mathcal{CS}^{j}_t$ be the convex safe set at iteration $j$ as defined in \eqref{eq:CS}. Let Assumptions~\ref{ass:feasTrajGiven}-\ref{ass:convexity} hold and assume that $x_0^j = x_S ~\forall j\geq 0$, then at every iteration $j\geq1$ the relaxed LMPC~\eqref{eq:RelaxedFTOCP} and \eqref{eq:policyLMPCRelaxed} is feasible for all $t \geq 0$ when \eqref{eq:policyLMPCRelaxed} is applied to system~\eqref{eq:system}.
\end{theorem}
\begin{proof}
We notice that by Assumption~\ref{ass:simplified} it follows that $\forall x \in \mathcal{CS}^{j}_t$ there exists $u \in \mathcal{U}$ such that $f(x,u) \in \mathcal{CS}^{j}_{t+1}$. Therefore, the recursive feasibility property follows from standard MPC arguments~\cite{BorrelliBemporadMorari_book, mayne2000constrained}.
\end{proof}

\subsection{Convergence and Performance Improvement}
We show that the closed-loop system \eqref{eq:system} and \eqref{eq:policyLMPC} converges in finite time to the terminal state $x_F$. 
Furthermore, the time $T^j$ at which the closed-loop system converges to the terminal state $x_F$ is non-increasing with the iteration index, i.e., $T^j \leq T^i, \forall i \in \{0, \ldots, j-1\}$. 
In the following, we present a side result which will be used in the main theorem. 
\begin{proposition}\label{prop:convNSteps}
Consider system \eqref{eq:system} controlled by the \mbox{LMPC}~\eqref{eq:FTOCP} and \eqref{eq:policyLMPC}. Assume that $\mathcal{SS}^{j-1}_t = x_F$ and $Q^{j-l}_{t}=0$ for all $t\geq0$. If at time $t$ Problem~\eqref{eq:FTOCP} is feasible, then the closed-loop system~\eqref{eq:system} and \eqref{eq:policyLMPC} converges in at most $t+N$ time steps to $x_F$.
\end{proposition}
\begin{proof}
By assumption, Problem~\eqref{eq:FTOCP} is feasible at time $t$ and there exists a sequence of feasible inputs which steers the system from $x_t^j$ to $x_F$ in at most $N+1$ steps. Therefore, we have that
\begin{equation}\label{eq:costBound}
    J_{t\rightarrow t+N}^{\scalebox{0.4}{LMPC},j}(x_t^j) \leq N.
\end{equation}
Furthermore, as $\mathcal{SS}^{j-1}_t = x_F$ is an invariant and $Q^{j-l}_{t}=0$, we have that the \mbox{LMPC}~\eqref{eq:FTOCP} and \eqref{eq:policyLMPC} is feasible at all time instants and
\begin{equation}\label{eq:LyDec}
    J_{t\rightarrow t+N}^{\scalebox{0.4}{LMPC},j}(x_t^j) \geq \mathds{1}(x_t) + J_{t+1\rightarrow t+1+N}^{\scalebox{0.4}{LMPC},j}(x_t^j).
\end{equation}
Now, we assume that $x_{i} \neq x_F \forall i \in \{t, \ldots, t+N-1\}$ . Therefore by~\eqref{eq:costBound}-\eqref{eq:LyDec} we have that at time $k = t+N-1$
\begin{equation*}
    J_{k \rightarrow k +N}^{\scalebox{0.4}{LMPC},j}(x_{k}^j) \leq J_{t\rightarrow t+N}^{\scalebox{0.4}{LMPC},j}(x_t^j) - \sum_{i=t}^{k} \mathds{1}(x_i) \leq N - (N - 1) = 1
\end{equation*}
which implies that $x_{ k+1}=x_{t+N} = x_{t+N|t+N-1}^{j,*} = x_F$.
\end{proof}

\begin{theorem}\label{th:convergence}
Consider system \eqref{eq:system} controlled by the LMPC~\eqref{eq:FTOCP} and \eqref{eq:policyLMPC}.
Let $\mathcal{SS}^{j}_t$ be the time varying safe set at iteration $j$ as defined in \eqref{eq:SS}. Let Assumption~\ref{ass:feasTrajGiven} hold and assume that $x_0^j = x_S$ and $T^j>N$, $\forall j\geq 0$. Then the time $T^j$ at which the closed-loop system \eqref{eq:system} and \eqref{eq:policyLMPC} converges to $x_F$ is non-increasing with the iteration index,
\begin{equation*}
    T^j \leq T^k,~\forall k\in \{0,\ldots, j-1\}.
\end{equation*}
\end{theorem}
\begin{proof}
By Theorem~\ref{th:recFeas} we have that the LMPC~\eqref{eq:FTOCP} and \eqref{eq:policyLMPC} is feasible for all time $t\geq0$. Denote 
\begin{equation*}
    T^{j-1,*} = \min_{k \in \{0,\ldots,j-1 \}} T^k
\end{equation*} 
as the minimum time to complete the task associated with the trajectories used to construct $\mathcal{SS}^{j-1}_{t+N}$. By definitions~\eqref{eq:SS}-\eqref{eq:deltaTime}, we have that at time $\bar t = T^{j-1,*} - N$ 
\begin{equation*}
    SS^{j-1}_{\bar t + N} = SS^{j-1}_{T^{j-1,*}} = x_F.
\end{equation*}
Therefore, by Proposition~\ref{prop:convNSteps} the closed-loop system converges in at most $\bar t + N = T^{j-1,*} $ time steps. Finally, we notice that $T^j = \bar t + N = T^{j-1,*} \leq T^k,~\forall k\in \{0,\ldots, j-1\}$.
\end{proof}

Next, we show that if the relaxed LMPC~\eqref{eq:RelaxedFTOCP} and \eqref{eq:policyLMPCRelaxed} is in closed-loop with system~\eqref{eq:system} which satisfies Assumption~\ref{ass:simplified}, then $T^j$ is non-increasing with the iteration index. The proof follows as in Theorem~\ref{th:convergence} leveraging the recursive feasibility of the relaxed LMPC~\eqref{eq:RelaxedFTOCP} and \eqref{eq:policyLMPCRelaxed} from Theorem~\ref{th:recFeasSimple}.

\begin{proposition}\label{prop:convNStepsSimplified}
Consider system \eqref{eq:system} controlled by the LMPC~\eqref{eq:RelaxedFTOCP} and \eqref{eq:policyLMPCRelaxed}. Assume that $\mathcal{CS}^{j-1}_t = x_F$ and $\bar Q^{j-l}_{t}=0$ for all $t\geq0$. If at time $t$ Problem~\eqref{eq:RelaxedFTOCP} is feasible, then the closed-loop system~\eqref{eq:system} and~\eqref{eq:policyLMPCRelaxed} converges in at most $t+N$ time steps to $x_F$.
\end{proposition}
\begin{proof}
The proof follows as in Proposition~\ref{prop:convNSteps} replacing the LMPC cost $J_{t\rightarrow t+N}^{\scalebox{0.4}{LMPC},j}(\cdot)$ with the relaxed LMPC cost $\bar J_{t\rightarrow t+N}^{\scalebox{0.4}{LMPC},j}(\cdot)$.
\end{proof}
\begin{theorem}\label{th:convergenceSimple}
Consider system \eqref{eq:system} controlled by the \mbox{LMPC} \eqref{eq:RelaxedFTOCP} and \eqref{eq:policyLMPCRelaxed}.
Let $\mathcal{CS}^{j}_t$ be the time varying safe set at iteration $j$ as defined in \eqref{eq:CS}. Let Assumptions~\ref{ass:feasTrajGiven}-\ref{ass:convexity} hold and assume that $x_0^j = x_S$ and $T^j >N$, $\forall j\geq 0$. Then the time $T^j$ at which the closed-loop system \eqref{eq:system} and \eqref{eq:policyLMPCRelaxed} converges to $x_F$ is non-increasing with the iteration index,
\begin{equation*}
    T^j \leq T^k,~\forall k\in \{0,\ldots, j-1\}.
\end{equation*}
\end{theorem}
\begin{proof}
By Theorem~\ref{th:recFeasSimple} we have that Problem~\eqref{eq:RelaxedFTOCP} is feasible at all time $t\geq0$. Therefore, the proof follows as for Theorem~\ref{th:convergence} using Proposition~\ref{prop:convNStepsSimplified}.
\end{proof}

\subsection{Sufficient Condition for the Relaxed LMPC}
In the previous sections, we discussed the properties of the relaxed LMPC strategy in closed-loop with nonlinear systems which satisfy Assumption~\ref{ass:simplified}. Next, we show that the recursive constraint satisfaction and non-decreasing performance properties still hold, if we replace Assumption~\ref{ass:simplified} with the following assumption on the system dynamics and stored data.

\begin{assumption}\label{property:storedDataAndDynamics}
Consider a convex safe set $\mathcal{CS}^j_t$ constructed using the stored closed-loop trajectories ${\bf{x}}^i$ and input sequences ${\bf{u}}^i$ for $i \in \{0,\ldots,j\}$. 
For all $k \in \{1,\ldots, n+1\}$, $x^{(k)} \in \big\{ \bigcup_{i=0}^j \bigcup_{t=0}^{T^i} x_t^i \big\}$ and 
$x \in \text{Conv}\big(\bigcup_{k=1}^{n+1}  x^{(k)}\big)$,
we have that there exists an input $u \in \mathcal{U}$ such that
\begin{equation*}
    \textstyle f(x,u) \in \text{Conv}\Big( \bigcup_{k=1}^{n+1} f(x^{(k)}, u^{(k)})\Big).
\end{equation*}
where $u^{(k)}$ is the stored input applied at the stored state $x^{(k)} \in \big\{ \bigcup_{i=0}^j \bigcup_{t=0}^{T^i} x_t^i \big\}$.

\end{assumption}

\medskip

The above assumption implies that given a state $x$ which can be expresses as the convex combination of $n+1$ stored states $\{x^{(1)},\ldots,x^{(n+1)}\}$ used to construct the convex safe set, there exists a control action $u \in \mathcal{U}$ which keeps the evolution of the system within the convex hull of the successor states $\{f(x^{(1)}, u^{(1)}),\ldots,f(x^{(n+1)}, u^{(n+1)})\}$. We underline that the above assumption is hard to verify in general. In practice, Assumption~\ref{property:storedDataAndDynamics} may be approximately checked using sampling strategies, as shown in the result section. 

Finally, we state the following theorem which summaries the sufficient conditions that guarantee recursive constraint satisfaction, convergence in finite time and non-decreasing performance at each iteration for the relaxed LMPC in closed-loop with the nonlinear system~\eqref{eq:system}.
\begin{theorem}
Consider system \eqref{eq:system} controlled by the relaxed LMPC~\eqref{eq:RelaxedFTOCP} and \eqref{eq:policyLMPCRelaxed}.
Let $\mathcal{CS}^{j}_t$ be the time varying convex safe set at iteration $j$ as defined in \eqref{eq:CS}. Let Assumptions~\ref{ass:feasTrajGiven}, \ref{ass:convexity} and \ref{property:storedDataAndDynamics} hold and assume that $x_0^j = x_S$ and $T^j > N$, $\forall j\geq 0$.
Then, the relaxed LMPC~\eqref{eq:RelaxedFTOCP} and \eqref{eq:policyLMPCRelaxed} satisfies state and input constraints~\eqref{eq:stateInputConstr} at all time. Furthermore, the time $T^j$ at which the closed-loop system~\eqref{eq:system} and~\eqref{eq:policyLMPCRelaxed} converges to $x_F$ is non-increasing with the iteration index,
\begin{equation*}
    T^j \leq T^k,~\forall k\in \{0,\ldots, j-1\}.
\end{equation*}
\end{theorem}
\begin{proof}
We assume that at time $t$ the relaxed LMPC~\eqref{eq:RelaxedFTOCP} and \eqref{eq:policyLMPCRelaxed} is feasible, let~\eqref{eq:optSolution} be the optimal solution. As Assumption~\ref{property:storedDataAndDynamics} holds, we have that there exists $u \in \mathcal{U}$ such that
\begin{equation*}
\begin{aligned}
    &[x_{t+1|t}^{j,*}, \ldots, x_{t+N|t}^{j,*},f(x_{t+N|t}^{j,*}, u) \in \mathcal{CS}^{j}_{t+1}] \\
    &[u_{t+1|t}^{j,*}, \ldots, u_{t+N-1|t}^{j,*},u \in \mathcal{U}],
\end{aligned}
\end{equation*}
satisfy state and input constraints~\eqref{eq:stateInputConstr}, and therefore the relaxed LMPC~\eqref{eq:RelaxedFTOCP} and \eqref{eq:policyLMPCRelaxed} is feasible at time $t+1$. The rest of the proof follows as in Theorems~\ref{th:recFeasSimple} and \ref{th:convergenceSimple}.
\end{proof}

\section{Data Reduction}\label{sec:dataReduction}
In this section, we show that the proposed LMPC can be implemented using a subset of the time varying safe set from~\eqref{eq:SS}. 
In particular, we show that the controller may be implemented using the last $l$ iterations and $P$ data points per iteration.

\subsection{Safe Subset}
We define the time varying safe subset from iteration $l$ to iteration $j$ and for $P$ data points as
\begin{equation}\label{eq:SS^j,l}
    \mathcal{SS}^{j,l}_{t,P} = \bigcup_{i = l}^j \bigcup_{k = \delta^{j,i}_t}^{\delta^{j,i}_t+P} x_{k}^i,
\end{equation}
where $\delta^{j,i}_t$ is defined in~\eqref{eq:deltaTime}. Furthermore, in the above definition we set $x_{k}^i = x_F$ for all $k>T^i$ and $i\in \{l,\ldots, j\}$. 
A representation of the time varying safe subset for a two-dimensional system is shown in Figure~\ref{fig:localSSillustration}. 
Compare the safe subset $\mathcal{SS}^{j,l}_{t,P}$ with the safe set $\mathcal{SS}^{j}_{t}$ from~\eqref{eq:SS}. We notice that,  $\mathcal{SS}^{j,l}_{t,P}$ is contained within $\mathcal{SS}^{j}_{t}$. Therefore, at time $t$ the safe subset
collects the stored states from which system~\eqref{eq:system} can reach the terminal state $x_F$ in at most $(T^{j,*}-t)$ time steps.
Finally, by definition, if a state $x_t^i$ belongs to $\mathcal{SS}^{j,l}_{t,P}$, then there exists a feasible control action $u_t^i \in \mathcal{U}$ which keeps the evolution of the nonlinear system~\eqref{eq:system} within the time varying safe set at the next time step $t+1$. 
This property allows us to replace $\mathcal{SS}^{j}_{t}$ with $\mathcal{SS}^{j,l}_{t,P}$ in the design of the LMPC~\eqref{eq:FTOCP} and~\eqref{eq:policyLMPC}, without loosing the recursive constraint satisfaction property from Theorem~\ref{th:recFeas}.

Finally, at each time $t$ we define the local convex safe subset as the convex hull of $\mathcal{SS}^{j,l}_{t,P}$ from~\eqref{eq:SS^j,l},
\begin{equation}\label{eq:CS^j,l}
\begin{aligned}
\mathcal{CS}^{j,l}_{t,P} &= \text{Conv}\big( \mathcal{SS}^{j,l}_{t,P} \big).
\end{aligned}
\end{equation}
We underline that $\mathcal{CS}^{j,l}_{t,P} \subseteq \mathcal{CS}^{j}_{t}$ and that the relaxed LMPC from Section~\ref{sec:LMPC_Nonlinear} may be implemented replacing the convex safe set~\eqref{eq:CS} with the convex safe subset~\eqref{eq:CS^j,l}. 

\begin{figure}[h!]
	\centering \includegraphics[width=0.5\columnwidth]{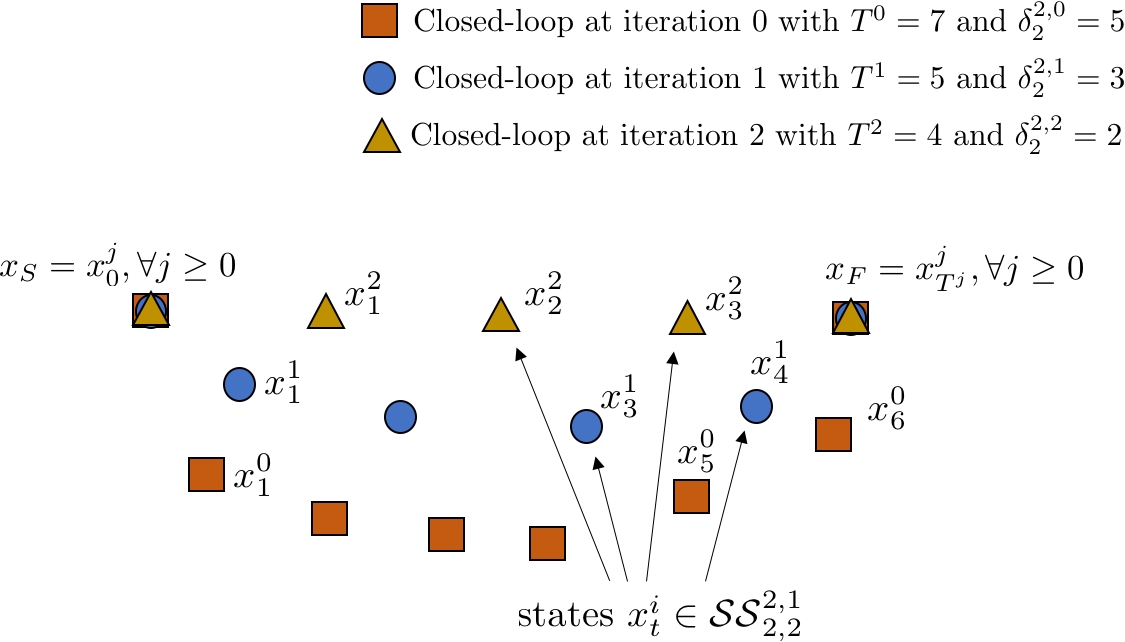}
	\caption{Representation of the time varying safe subset $\mathcal{SS}^{2,1}_{2,2}$. 
	Compared with the safe set from Figure~\ref{fig:SSillustration}, the safe subset is constructed using a subset of data points. In particular, we have that $\mathcal{SS}^{2}_{2} = \{x_2^2, x_3^2, x_4^2,x_3^1,x_4^1,x_5^1,x_5^0,x_6^0,x_7^0 \}$ and $\mathcal{SS}^{2,1}_{2,2} = \{x_2^2, x_3^2,x_3^1,x_4^1 \}$.}\label{fig:localSSillustration}
\end{figure}

\subsection{Q-function}
In the section, we construct the $Q$-function which assigns the cost-to-go to the states contained in the time varying safe subset from~\eqref{eq:SS^j,l}. In particular, we introduce the function $Q^{j,l}_{t,P}(\cdot)$, defined over the safe subset $\mathcal{SS}^{j,l}_{t,P}$, as
\begin{equation}\label{eq:Qfunction^j,l}
\begin{aligned}
Q^{j,l}_{t,P}(x) = \min\limits_{\substack{i \in \{l, \ldots, j\} \\ t \in \{\delta_t^{j,i}, \ldots, \delta_t^{j,i}+P\} } } & \quad J^i_{t\rightarrow T^i}(x_t^i)\\
\text{s.t. }\quad  \quad & \quad x = x_t^i \in \mathcal{SS}^{j,l}_{t,P}.
\end{aligned}
\end{equation}
Compare the above function $Q^{j,l}_{t,P}$ with $Q^{j}_{t}$ from~\eqref{eq:Qfunction}. We notice that, the domain of $Q^{j,l}_{t,P}$ is the safe subset $\mathcal{SS}^{j,l}_{t,P}$ and the domain of the $Q^{j}_{t}$ is the safe set $\mathcal{SS}^{j}_{t} \supseteq \mathcal{SS}^{j,l}_{t,P}$. Moreover, we have that
\begin{equation*}
    \forall x \in \mathcal{SS}^{j,l}_{t,P}, Q^{j,l}_{t,P}(x) = Q^{j}_{t}(x).
\end{equation*}
The above property allows us to replace $Q^{j}_{t}$ with $Q^{j,l}_{t,P}$ in the design of the LMPC policy~\eqref{eq:policyLMPC}, without loosing the finite time convergence and non-decreasing performance properties.

Furthermore, we define the convex $Q$-function $\bar{Q}^{j,l}_{t,P}$ from iteration $l$ to iteration $j$ and for $P$ data points as
\begin{equation}\label{eq:cvxQfunction^j,l}
\begin{aligned}
\bar{Q}^{j,l}_{t,P}(x) = \min\limits_{ [\lambda_{\delta_t^{j,i}}^0, \ldots, \lambda_{\delta_t^{j,i}+P}^j] \geq 0 } & \quad \sum_{i = 0}^j \sum_{k = \delta_t^{j,i}}^{\delta_t^{j,i}+P} \lambda_{k}^i J^i_{k\rightarrow \delta_t^{j,i}+P}(x_{k}^{i})\\
\text{s.t. }\quad\quad & \quad \sum_{i = 0}^j \sum_{k = \delta_t^{j,i}}^{\delta_t^{j,i}+P} \lambda_{k}^i x_k^i = x \\
& \quad \sum_{i = 0}^j \sum_{k = \delta_t^{j,i}}^{\delta_t^{j,i}+P} \lambda_{k}^i = 1.
\end{aligned}
\end{equation}
where $\delta^{j,i}_t$ is defined in~\eqref{eq:deltaTime}. 
The above convex $Q$-function $\bar{Q}^{j}_t(\cdot)$ is simply a piecewise-affine interpolation of the $Q$-function from~\eqref{eq:Qfunction^j,l} over the convex safe subset, as shown in Figure~\ref{fig:localConvexQfunctionillustration}. In the result section we will show that $\bar{Q}^{j,l}_{t,P}$ can be used in the relaxed LMPC design instead of $\bar{Q}^{j}_{t}$. 

\begin{figure}[h!]
	\centering \includegraphics[width=0.5\columnwidth]{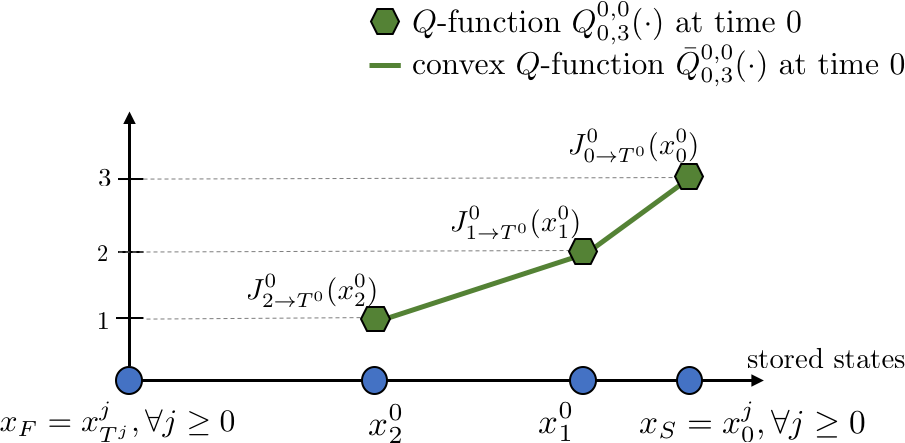}
	\caption{Representation of the $Q$-function $Q^{0,0}_{0,3}(\cdot)$ and convex $Q$-function $\bar Q^{0,0}_{0,3}(\cdot)$. We notice that the $Q$-function $Q^{0,0}_{0,3}(\cdot)$ is defined over a set of discrete data points, whereas the convex $Q$-function $\bar Q^{0,0}_{0,3}(\cdot)$ is defined over the convex safe set.}\label{fig:localConvexQfunctionillustration}
\end{figure}


\section{Beyond Iterative Tasks and Deterministic Systems}
In this section, we describe how the proposed strategy can be used when the initial condition is perturbed at each iteration and when the system dynamics are subject to bounded disturbances. In particular, we use backward reachablity analysis to characterize the region of attraction of the controller. Furthermore, we use standard rigid tube MPC strategies to extend the control design to uncertain systems.

\subsection{Perturbed Initial Condition}\label{sec:perturbed}
In this section, we assume that the initial condition $x_0^j$ may be perturbed at each iteration. First, we introduce the one-step controllable set and the $k$-steps controllable set to a set $\mathcal{S}$ from~\cite{BorrelliBemporadMorari_book}.

\begin{definition}[One-Step Controllable Set]
For the system~\eqref{eq:system} we denote the \emph{one-step controllable set} to  the set $\mathcal{S}$ as
\end{definition}
\begin{equation}\label{eq:Pred}
\begin{aligned}
\mathcal{K}_1(\mathcal{S}) = \Pre(\mathcal{S}) \cap\mathcal{X}.
\end{aligned}
\end{equation}
where
\begin{equation}\label{eq:Pred0}
\begin{aligned}
\Pre(\mathcal{S})\triangleq \{x\in \mathbb{R}^n~:~\exists u\in\mathcal{U}\text{ s.t. } f(x,u)\in \mathcal{S}\}.
\end{aligned}
\end{equation}
$\mathcal{K}_1(\mathcal{S})$ is the set of states which can be driven to the target set $\mathcal{S}$ in one time step while satisfying input and state constraints. $N$-step controllable sets are defined by iterating $\mathcal{K}_1(\mathcal{S})$  computations.

\begin{definition}[$N$-Step Controllable Set $\tKN(\mathcal{S})$] \index{$N$-Step Controllable Set}
\label{def:Ncset}
For a given target set $\mathcal{S}\subseteq \mathcal{X}$, the $N$-step controllable set $\tKN(\mathcal{S})$ of system~\eqref{eq:system} subject to constraints~\eqref{eq:stateInputConstr} is defined recursively as:
\begin{equation}
\label{eq:Nsteprecf}
\mathcal{K}_j(\mathcal{S}) \triangleq \Pre(\mathcal{K}_{j-1}(\mathcal{S}))\cap\mathcal{X},~~ \mathcal{K}_{0}(\mathcal{S})=\mathcal{S},~~~~ j\in\{1,\ldots,N\}.
\end{equation}
\end{definition}
From Definition~\ref{def:Ncset}, all states $x_0$ of the system \eqref{eq:system} belonging to the $N$-step controllable set $\tKN(\mathcal{S})$ can be driven, by a suitable control sequence, to the target set $\mathcal{S}$ in $N$ steps,
while satisfying input and state constraints. Therefore, if the initial state $x_0^j \neq x_S$ belongs to the $N$-step controllable set $\tKN(\mathcal{SS}^{j-1}_{N})$, then we have that Problem~\eqref{eq:FTOCP} is feasible and the LMPC properties hold for the system initialized at $x_0^j \neq x_S$, as stated by the following theorem.

\begin{theorem}\label{th:diff_IC}
Consider $x_0^j \neq x_S$ and system \eqref{eq:system} controlled by the \mbox{LMPC} \eqref{eq:FTOCP} and \eqref{eq:policyLMPC}.
Let $\mathcal{SS}^{j}_t$ be the time varying safe set at iteration $j$ as defined in \eqref{eq:SS}. Let Assumption~\ref{ass:feasTrajGiven} hold and assume that $x_0^j \in \tKN(\mathcal{SS}_{N}^{j-1}) ~\forall j\geq 1$. Then at every iteration $j\geq1$ the LMPC~\eqref{eq:FTOCP} and \eqref{eq:policyLMPC} is feasible for all $t \geq 0$ when \eqref{eq:policyLMPC} is applied to system~\eqref{eq:system}. Furthermore, the time $T^j$ at which the closed-loop system~\eqref{eq:system} and~\eqref{eq:policyLMPC} converges to $x_F$ is non-increasing with the iteration index,
\begin{equation*}
    T^j \leq T^k,~\forall k \in \{0,\ldots, j-1\}.
\end{equation*}
\end{theorem}

\begin{proof}
We notice that by assumption $x_0^j \in \tKN(\mathcal{SS}_{t+N}^{j-1}) ~\forall j\geq 1$. Therefore, by definition of $k$-steps reachable set the LMPC~\eqref{eq:FTOCP} and \eqref{eq:policyLMPC} is feasible at time $t=0$. We notice that, as discussed in Section~\ref{sec:SSandVfun}, $\forall x \in \mathcal{SS}_{N}^{j-1}$ we have that $Q^{j-1}_N(x) \leq T^{j-1,*}-N$. Therefore, the LMPC optimal cost at time $t=0$ is $J_{0\rightarrow N}^{\scalebox{0.4}{LMPC},j}(x_0^j) \leq T^{j-1,*}$. The rest of the proof follows as in Theorems~\ref{th:recFeas} and~\ref{th:convergence}.
\end{proof}

Finally, we underlined that the guarantees from the above theorem hold also for the relaxed LMPC from Section~\ref{sec:LMPC_Nonlinear}, when $\tKN(\mathcal{SS}_{N}^j)$ is replaced with $\tKN(\mathcal{CS}_{N}^j)$.

\subsection{Uncertain Systems}
All guarantees provided in this paper hold for deterministic models without uncertainty. In this section, we briefly show how the proposed strategies can be combined with standard rigid tube MPC methodologies to design a robust LMPC for uncertain systems. We consider the following nonlinear system
\begin{equation}\label{eq:uncertainsSys}
    x_{k+1}^j = f_w(x_t^j, u_t^j) + w_t^j,
\end{equation}
where at time $t$ of iteration $j$ the state $x_t^j \in \mathbb{R}^n$, the input $u_t^j \in \mathbb{R}^d$ and the disturbance $w_t^j \in \mathcal{W}$. Furthermore, we introduce the nominal state $\bar x \in \mathbb{R}^n$, the error state $e = x-\bar x \in \mathbb{R}^n$ and the associated dynamics 
\begin{equation}\label{eq:decoulpledSys}
\begin{aligned}
    \bar x_{t+1}^j &= f_w(\bar x_t^j, \bar u_t^j) \\
    e_{t+1}^j &= f_w(x_t^j, u_t^j) + w_t^j - f_w(\bar x_t^j, \bar u_t^j),
\end{aligned}
\end{equation}
where $\bar u_t^j \in \mathbb{R}^d$ represents the nominal input. 
The above decomposition has been used in several robust MPC and motion planning strategies~\cite{singh2017robust, singh2018robust, yin2019optimization, yu2013tube, herbert2017fastrack}. In these approaches, the key idea is to compute a robust control invariant set $\mathcal{E}$ for the error dynamic and then use the nominal model for planning. The robust invariant set and the associated control policy for the error dynamics may be computed using sum of square programming~\cite{singh2017robust,
singh2018robust, yin2019optimization}, Lipschitz properties of the nonlinear dynamics~\cite{yu2013tube} or Hamilton-Jacobi reachability analysis~\cite{herbert2017fastrack}. In the following we assume that a robust invariant set $\mathcal{E}$ for the error dynamics is given.

\begin{assumption}\label{ass:invariance}
Consider the uncertain system~\eqref{eq:uncertainsSys} and the constraint sets~\eqref{eq:stateInputConstr}.
For the set $\mathcal{E} \subset \mathcal{X}$ and the policy $\kappa_w:\mathbb{R}^n\times\mathbb{R}^n \rightarrow \mathcal{V} \subset \mathcal{U}$ we have that
\begin{equation*}
    \forall x \in \mathcal{X}, \forall e =x - \bar x \in \mathcal{E},  f_w(x, \bar u + \kappa_w(x,\bar x)) + w - f_w(\bar x, \bar u) \in \mathcal{E}, \forall w \in \mathcal{W}, \forall \bar u \in \mathcal{U} \ominus \mathcal{V}
\end{equation*}
where $\mathcal{U} \ominus \mathcal{V}$ denotes the Pontryagin difference between the sets $\mathcal{U}$ and $\mathcal{V}$.
\end{assumption}

Given $j$ stored trajectories for the nominal system from~\eqref{eq:decoulpledSys}, 
\begin{equation}\label{eq:storedNominal}
{\bar {\bf{x}}}^i = [\bar x_0^i,\ldots,~\bar x_{T^i}^i] \text{ for }i\in\{0,\ldots,j\},
\end{equation}
we define the nominal time varying safe set at iteration $j$ as
\begin{equation}\label{eq:nominalSS}
    \mathcal{\bar{SS}}^{j}_t = \bigcup_{i = 0}^j \bigcup_{k = \delta^{j,i}_t}^{T^i} \bar x_{k}^i,
\end{equation}
where $T^{j,*} = \min_{k \in \{ 0, \ldots, j \}}  T^k $
and $\delta^{j,i}_t $ is defined as in~\eqref{eq:deltaTime}. Furthermore, we introduce the nominal $Q$-function
\begin{equation}\label{eq:nominalQfunction}
\begin{aligned}
\bar Q^{j}_t(x) = \min\limits_{\substack{i \in \{0, \ldots, j\} \\ k \in \{\delta_t^{j,i}, \ldots, T^i\} } } & \quad J^i_{k\rightarrow T^i}(\bar x_k^i)\\
\text{s.t. } \quad~ & \quad x = \bar x_k^i \in \mathcal{\bar{SS}}^{j}_t.
\end{aligned}
\end{equation}
where $J_{t\rightarrow T^j}^j(\bar x_t^j) = ~ \sum_{k=t}^{T^j} \mathds{1}_{x_F}(\bar x_k^j)$. The above nominal $Q$-function maps each state $\bar x_t^j$ of the nominal safe set to the minimum cost-to-go along the stored nominal trajectories~\eqref{eq:storedNominal}.

\begin{remark}
The nominal safe set $\mathcal{\bar{SS}}_t^j$~\eqref{eq:nominalSS} and nominal $Q$-function $\bar Q^{j}_t(x)$~\eqref{eq:nominalQfunction} are defined similarly to the safe set and $Q$-function described in Section~\ref{sec:SSandVfun}. The difference between the deterministic case and the uncertain one is that the nominal safe set and $Q$-function are constructed using the nominal stored trajectories from~\eqref{eq:storedNominal}.
\end{remark}

Next, we show how to leverage the nominal safe set and $Q$-function to design a robust LMPC, which iteratively steers the uncertain system~\eqref{eq:uncertainsSys} from the starting state $x_S$ to a goal set $\mathcal{G} = \{x_F\} \oplus \mathcal{E}$. Consider the following optimal control problem,
\begin{equation}\label{eq:robustFTOCP}
	\begin{aligned}
		J_{t\rightarrow t+N}^{\scalebox{0.4}{RLMPC},j}(x_t^j, \bar x_{t-1}^j) = \min_{\bar x_{t|t}^j, {\bf{U}}_t^j } \quad & \bigg[  \sum_{k=t}^{t+N-1}  \mathds{1}_{x_F}(\bar x_{k|t}^j) + \bar{Q}_{t+N}^{j-1}(\bar x_{t+N|t}^j)\bigg] \\
		\text{s.t.}~\quad 
		&\bar x_{k+1|t}^j=f_w( \bar x_{k|t}^j, \bar u_{k|t}^j ), \forall k = t, \cdots, t+N-1 \\
		&\bar x_{k|t}^j \in \mathcal{X}\ominus\mathcal{E}, \bar u_{k|t}^j \in \mathcal{U}\ominus\mathcal{V}, \forall k = t, \cdots, t+N-1\\ 
		& \bar x_{t+N|t}^j \in ~\mathcal{\bar{SS}}_{t+N}^{j-1}\\
		& x_t^j - \bar x_{t|t}^j \in \mathcal{E} \\
		&\begin{rcases}
  \bar x_{t|t}^j = f_w(\bar x_{t-1}^j, \bar u_{t-1|t}^j) \in \mathcal{X} \ominus \mathcal{E}\\ \bar u_{t-1|t}^j \in \mathcal{U}\ominus\mathcal{V}
\end{rcases}
\text{If } t \geq 1
	\end{aligned}
\end{equation}
where ${\bf{U}}_t^j = [u_{t-1|t}, \ldots, u_{t+N-1|t}]$. The above finite time optimal control problem finds an initial state $\bar x_{t|t}^j$ and plans a trajectory which steers the nominal model to the nominal time varying safe set $\mathcal{\bar{SS}}_{t+N}^{j-1}$ from~\eqref{eq:nominalSS}. Notice that the state and input constraint sets in~\eqref{eq:robustFTOCP} are tightened to account for the model mismatch.  Let ${\bf{U}}_t^{j,*}$ and $\bar x_{t|t}^{j,*}$ be the optimal solution to the above finite time optimal control problem, then we apply to system~\eqref{eq:uncertainsSys}
\begin{equation}\label{eq:robustPolicy}
    u_t^j = \pi^{\scalebox{0.5}{RLMPC},j}_t(x_t^j, \bar x_{t-1}^j) = \bar u_t^{j,*} + \kappa_w(x_t^j, \bar x_{t|t}^{j,*}).
\end{equation}
Notice that the robust LMPC problem~\eqref{eq:robustFTOCP} differs from a standard fixed tube robust MPC for the following reasons: $(i)$ the sequence of predicted inputs has $N+1$ terms, $(ii)$ the nominal state $\bar x_{t-1}^j$ at the previous time step is used to constrain the nominal state $\bar x_{t|t}^j$ and $(iii)$ the last two constraints in problem~\eqref{eq:robustFTOCP} are removed at time $t=0$. These design choices guarantee that the nominal state and input trajectories
\begin{equation}\label{eq:nomStateInputSeq}
{\bar {\bf{x}}}^j = [\bar x_0^j,\ldots,~\bar x_{T^j}^j] \text{ where } \bar x_t^j = \bar x_{t|t}^{j,*}~\forall t\geq 0 \text{ and  } {\bar {\bf{u}}}^j = [\bar u_0^j,\ldots,~\bar u_{T^j}^j]  \text{ where } u_{t-1} = \bar u_{t-1|t}^{j,*}~\forall t\geq 1,
\end{equation}
are feasible for the nominal system~\eqref{eq:uncertainsSys}. Therefore, the above nominal trajectory could be used to update the nominal safe set~\eqref{eq:nominalSS} and nominal $Q$-function~\eqref{eq:nominalQfunction}.
It is important to underline that the nominal state trajectory in~\eqref{eq:nomStateInputSeq} is computed by the robust LMPC~\eqref{eq:robustFTOCP} and~\eqref{eq:robustPolicy} smoothing out the effect of the disturbance on the nominal dynamics. Indeed the controller can pick the nominal state $\bar x_t^j = \bar x_{t|t}^{j,*}$ as long as $e_t^j = x_t^j - \bar x_t^j \in \mathcal{E}$ and the nominal trajectory is feasible for some input $\bar u_{t-1}^j \in \mathcal{U} \ominus \mathcal{V}$.

In what follows, we show that the robust LMPC~\eqref{eq:robustFTOCP} and~\eqref{eq:robustPolicy} guarantees robust constraint satisfaction and non-decreasing performance for the closed-loop uncertain system~\eqref{eq:uncertainsSys} and \eqref{eq:robustPolicy}.

\begin{assumption}\label{ass:nominalfeasTrajGiven}
At iteration $j=0$, we are given the nominal closed-loop trajectory and associated input sequence
\begin{equation*}
\begin{aligned}
 [\bar x_0^0,\ldots,~\bar x_{T^0}^0]\text{ and } [\bar u_0^0,\ldots,~\bar u_{ T^0}^0], 
\end{aligned}
\end{equation*}
such that $x_t^0 \in \mathcal{X}\ominus\mathcal{E}$ and $\bar u_{t}^0 \in \mathcal{U} \ominus \mathcal{V}$, for all $t \in \{0, \ldots, T^0\}$. Furthermore, we have that $x_0^0 = \bar x_0^0 = x_S$ and $\bar x_{T^0}^0 = \bar x_F$, where $\bar x_F$ is an unforced equilibrium point for the nominal system~\eqref{eq:decoulpledSys}.
\end{assumption}

\begin{theorem}
Consider the uncertain system \eqref{eq:uncertainsSys} controlled by the robust \mbox{LMPC} \eqref{eq:robustFTOCP} and \eqref{eq:robustPolicy}.
Let $\mathcal{\bar{SS}}^{j}_t$ be the time varying safe set at iteration $j$ defined as in~\eqref{eq:nominalSS}. Let Assumptions~\ref{ass:invariance}-\ref{ass:nominalfeasTrajGiven} hold and $x_0^j \in \{x_s\}\oplus \mathcal{E},~\forall j\geq 1$. Then at every iteration $j\geq1$ the robust LMPC~\eqref{eq:robustFTOCP} and \eqref{eq:robustPolicy} is feasible for all $t \geq 0$ when \eqref{eq:robustPolicy} is applied to system~\eqref{eq:uncertainsSys} and state and input constraints~\eqref{eq:stateInputConstr} are robustly satisfied. Furthermore, the time $\bar T^j$ at which the  nominal state $\bar x_{T^j}^j$ equals the goal state $\bar x_F$ is non-increasing with the iteration index,
\begin{equation*}
    \bar T^j \leq \bar T^k,~\forall k \in \{0,\ldots, j-1\}.
\end{equation*}
Finally, at time $\bar T^j$ the uncertain system reaches the goal set $\mathcal{G} = \{x_F\} \oplus \mathcal{E}$.
\end{theorem}

\begin{proof}
We notice that at iteration $j$ the following nominal state and input sequences
\begin{equation*}
    \begin{aligned}
        [\bar x_0^{j-1},\ldots, \bar x_{N}^{j-1}] \text{ and }[0, \bar u_0^{j-1},\ldots, \bar u_{N-1}^{j-1}]\\
    \end{aligned}
\end{equation*}
are feasible for Problem~\eqref{eq:robustFTOCP} at time $t=0$, as $x_0^j \in \{x_s\}\oplus \mathcal{E}$ and the last two constraints in Problem~\eqref{eq:robustFTOCP} are not enforced at $t=0$. Now assume that at time $t\geq1$ Problem~\eqref{eq:robustFTOCP} is feasible. Let
\begin{equation*}
    \begin{aligned}
        [\bar x_{t|t}^{j,*},\ldots, \bar x_{t+N|t}^{j,*}] \text{ and }[\bar u_{t-1|t}^{j,*}, \bar u_{t|t}^{j,*},\ldots, \bar u_{t+N-1|t}^{j,*}]\\
    \end{aligned}
\end{equation*}
be the optimal state-input sequence, where $\bar x_{t+N|t}^{j,*} = \bar x_{k}^{i} \in \mathcal{\bar{SS}}^{j-1}_{t+N}$ for some $i\in \{0,\ldots, j\}$ and $k\in \{0,\ldots,T^i\}$. Then, from Assumption~\ref{ass:invariance} and~\eqref{eq:robustPolicy} the error $e_{t+1} = x_{t+1}-\bar x_{t+1|t}^{j,*} \in \mathcal{E}$ and therefore we have that
\begin{equation*}
    \begin{aligned}
        [\bar x_{t+1|t}^{j,*},\ldots, \bar x_{t+N|t}^{j,*}, \bar x_{k+1}^{i}] \text{ and }[\bar u_{t|t}^{j,*}, \bar u_{t+1|t}^{j,*},\ldots, \bar u_{t+N-1|t}^{j,*}, \bar u_{k}^{i}]\\
    \end{aligned}
\end{equation*}
is a feasible solution for Problem~\eqref{eq:robustFTOCP} at time $t+1$. Therefore, it follows from robust MPC arguments~\cite{chisci2001systems} that Problem~\eqref{eq:robustFTOCP} is feasible for all $t \geq 0$ and that state and input constraints~\eqref{eq:stateInputConstr} are robustly satisfied. The rest of the proof follows as in Theorem~\ref{th:convergence} by analysing the properties of the robust LMPC cost~$J_{t\rightarrow t+N}^{\scalebox{0.4}{RLMPC},j}(\cdot, \cdot)$.
\end{proof}

Finally, we underlined that the guarantees from the above theorem hold also for the relaxed LMPC from Section~\ref{sec:LMPC_Nonlinear} when in problem~\eqref{eq:FTOCP} the nominal safe set $\mathcal{\bar{SS}}_{t+N}^j$ is replaced with the nominal convex safe $\mathcal{\bar{CS}}_{N}^j$, which is computed as in Section~\ref{sec:SSandVfun} using the nominal trajectories from~\eqref{eq:nomStateInputSeq}.

\section{Results}
We test the proposed strategy on three time optimal control problems. In the first example, the LMPC is used to drive a dubins car from the starting point $x_S$ to the terminal point $x_F$ while avoiding the obstacle shown in Figure~\ref{fig:SSclosedLoopEvolution}. In the second example, we control a nonlinear double integrator system, which satisfies Assumption~\ref{ass:simplified}. Finally, the third example is a dubins car racing problem, which we solved using the relaxed LMPC after checking Assumption~\ref{property:storedDataAndDynamics} via sampling. The controller is implemented using CasADi~\cite{andersson2019casadi} for automatic differentiation and IPOPT~\cite{wachter2006implementation} to solve the nonlinear optimization problem. The code  is available at \url{https://github.com/urosolia/LMPC} in the {\fontfamily{qcr}\selectfont
NonlinearLMPC} folder.

\subsection{Minimum time obstacle avoidance}\label{sec:dubinsObs}
We use the LMPC policy from Section~\ref{sec:LMPC_MIP} on the minimum time obstacle avoidance optimal control problem from~\cite{rosolia2017learning},
\begin{equation*}\label{eq:dubinsObsAvoidance}
    \begin{aligned}
        \min_{\substack{T, a_0, \ldots, a_{T-1} \\  \theta_0, \ldots, \theta_{T-1}}} & \quad \sum_{t = 0}^{T-1} 1 \\
        \text{s.t. } \quad ~& \quad \begin{bmatrix}x_{t+1} \\ y_{t+1}\\ v_{t+1}\end{bmatrix}  = \begin{bmatrix}x_{t} + v_t \cos(\theta_t) \\ y_{t} + v_t\sin(\theta_t) \\ v_{t} + a_t \end{bmatrix}, \forall t \geq 0 \\
        & \quad \frac{(x_t - x_{\mathrm{obs}})^2}{a_x^2} + \frac{(y_t - y_{\mathrm{obs}})^2}{a_y^2} \geq 1, \forall t\geq 0 \\
        & \quad \begin{bmatrix} -\pi/2 \\ -1\end{bmatrix}\leq \begin{bmatrix} \theta_t \\ a_t\end{bmatrix} \leq \begin{bmatrix} \pi/2 \\ 1\end{bmatrix}, \forall t\geq 0 \\
        & \quad x_T = x_F =[54, 0,0]^T,\\
        & \quad x_0 = x_S = [0,0,0]^T.
    \end{aligned}
\end{equation*}
where $x_t$, $y_t$ and $v_t$ represent the position on the $X-Y$ plane and the velocity.
The goal of the controller is to steer the dubins car from the starting state $x_S$ to the terminal point $x_F$, while satisfying input saturation constraints and avoiding an obstacle. The obstacle is represented by an ellipse centered at $(x_{\mathrm{obs}}, y_{\mathrm{obs}})=(27,-1)$ with semi-axis $(a_x, a_y)=(8,6)$. At iteration $0$, we compute a first feasible trajectory using a brute force algorithm and we use the closed-loop data to initialize the LMPC~\eqref{eq:FTOCP} and \eqref{eq:policyLMPC} with $N=6$.

We compare the performance of the LMPC from~\cite{rosolia2017learning} and the LMPC policies~\eqref{eq:policyLMPC} synthesized using different number of data points $P=\{8,10,40\}$ and iterations $i=\{1,2,3\}$, as described in Section~\ref{sec:dataReduction} (in definition~\eqref{eq:SS^j,l} we set $l=j-1-i$). 
Figure~\ref{fig:SS_convergenceSpeed} shows the time steps $T^j$ at which the closed-loop system converged to $x_F$ at each iteration $j$. 
We notice that all LMPC policies converge to a steady state behavior which steers the system from $x_S$ to $x_F$ in $16$ time steps. 
Furthermore, Figure~\ref{fig:SS_convergenceSpeed} shows that the number of iterations needed to reach convergence is proportional to the amount of data used to synthesize the LMPC policy. 

\begin{figure}[h!]
	\centering \includegraphics[trim = 0mm 0mm 15mm 14mm, clip, width=0.5\columnwidth]{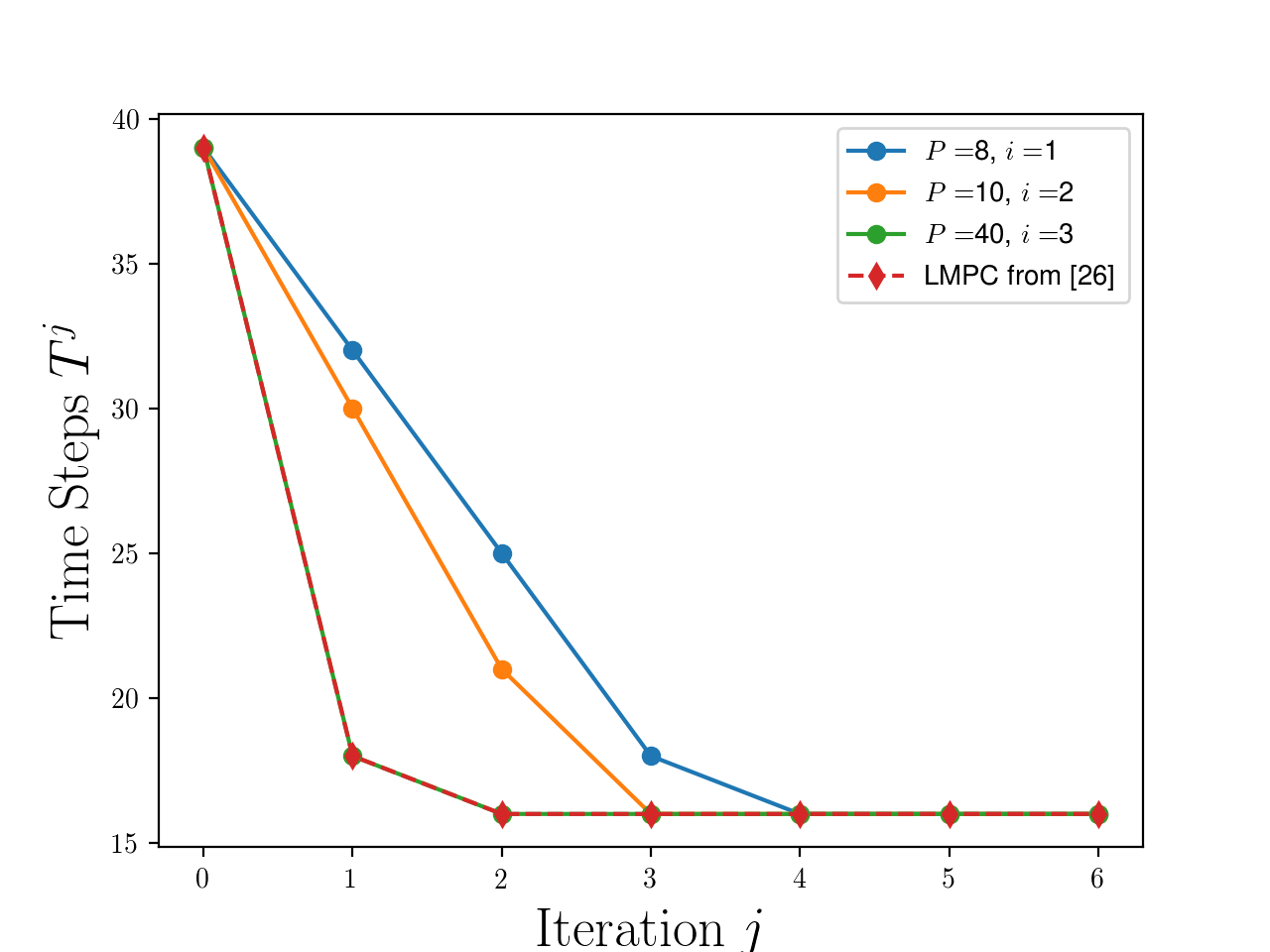}
	\caption{Time steps $T^j$ to reach $x_F$ as a function of the iteration index. We notice that as more data points are used in the synthesis process, the number of iterations needed to reach a steady state behavior decreases.}\label{fig:SS_convergenceSpeed}
\end{figure}

\begin{figure}[h!]
	\centering \includegraphics[trim = 0mm 0mm 15mm 14mm, clip, width=0.5\columnwidth]{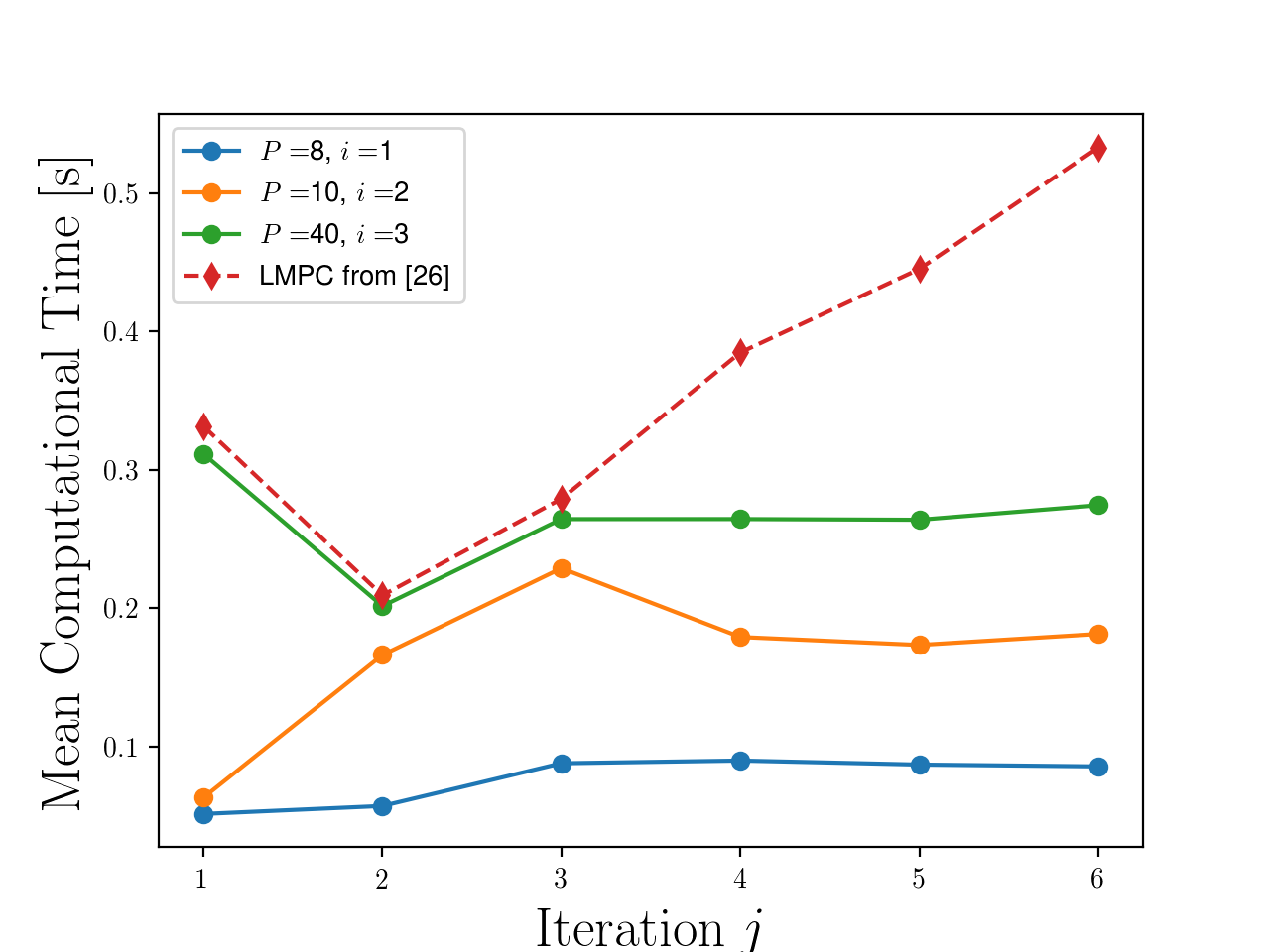}
	\caption{Computational cost associated with the LMPC policy at each time $t$ as function of the iteration index. We notice that as more data points are used in the synthesis process, the computational cost increases.}\label{fig:SScomputationalCost} 
\end{figure}

Figure~\ref{fig:SScomputationalCost} shows that the computational time increases as more data points $P$ are used in the control design. Therefore, there is a trade-off between the computational burden and the performance improvement shown in Figure~\ref{fig:SS_convergenceSpeed}.
Notice that, as the number of data points $P$ used for synthesis is constant, the computational cost associated with the proposed time varying LMPC strategy converges to a steady state value. On the other hand, the computation cost associate with the LMPC strategy from~\cite{rosolia2017learning} increases at each iteration.
Therefore, we confirm that the proposed time varying LMPC~\eqref{eq:FTOCP} and \eqref{eq:policyLMPC} enables the reduction of the computational cost while achieving the same closed-loop performance. We underline that we computed the solution to~\eqref{eq:FTOCP} by solving a set of nonlinear smooth optimization problems\footnote{Code available at \url{https://github.com/urosolia/LMPC} in the folder {\fontfamily{qcr}\selectfont
NonlinearLMPC/DubinsObstacleAvoidance\_SampleSafeSet}.}. At time $t$, for each of the $P(j-l)$ points stored in the safe subset~\eqref{eq:SS^j,l}, we solved a smooth nonlinear optimization problem. Afterwards, we selected the optimal solution associated with the minimum cost. Notice that the computational cost associated with the proposed strategy is proportional to the computational cost of a standard nonlinear MPC scaled by a factor $C = P(j-l)$, when parallel computing is not available.


\begin{figure}[h!]
	\centering \includegraphics[trim = 0mm 0mm 15mm 14mm, clip, width=0.5\columnwidth]{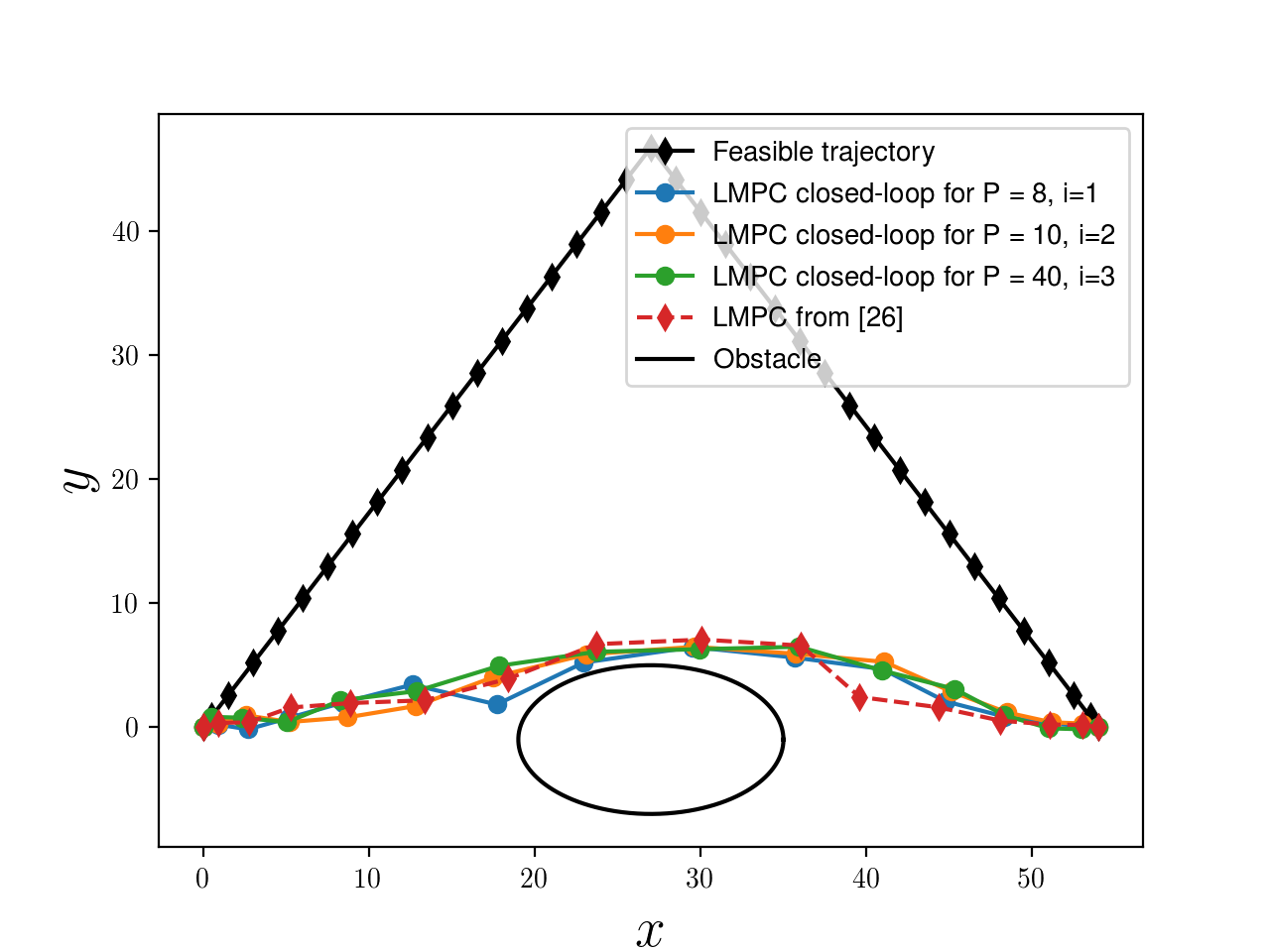}
	\caption{First feasible trajectory, stored data points and closed-loop trajectory at the $6$th iteration. We notice that the LMPC is able to avoid the obstacle at each iteration.}\label{fig:SSclosedLoopEvolution}
\end{figure}

\begin{figure}[h!]
	\centering \includegraphics[trim = 0mm 0mm 15mm 14mm, clip, width=0.5\columnwidth]{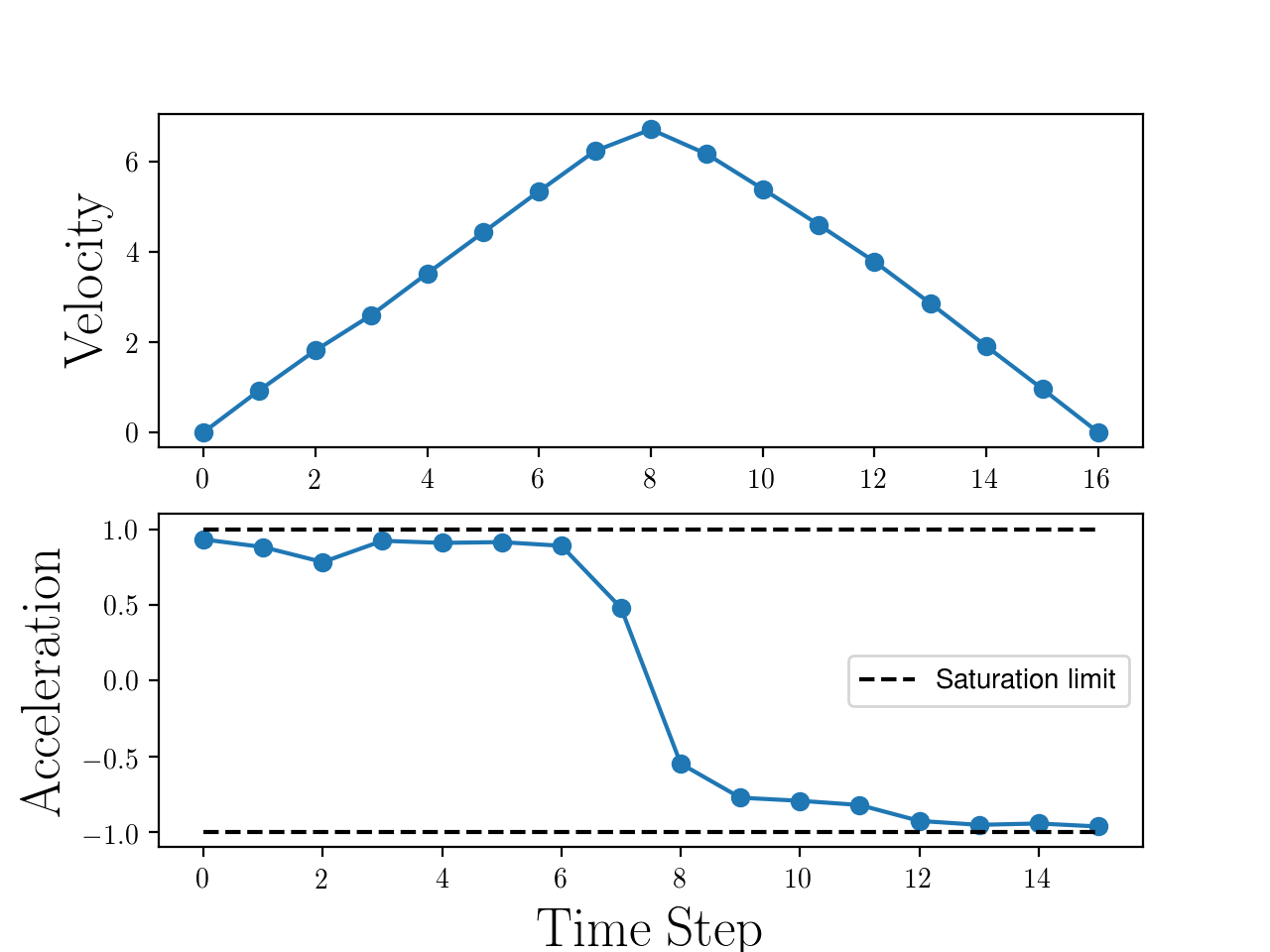}
	\caption{Acceleration and speed profile at convergence. We notice that the controller accelerates for the first $8$ time steps and afterwards it decelerates to reach the terminal goal state with zero velocity.}\label{fig:velocityAcceleration}
\end{figure}

Finally, we analyze the closed-loop trajectories associated with the LMPC policy~\eqref{eq:policyLMPC} synthesized with $P=8$ data points and $i=1$ iteration. Figure~\ref{fig:SSclosedLoopEvolution} shows the first feasible trajectory, the stored data points and the closed-loop trajectory at convergence. We confirm that the LMPC is able to explore the state space while avoiding the obstacle and steering the system from the starting state $x_S$ to the terminal state $x_F$. Furthermore, we notice that the LMPC accelerates during the first part of the task, and then it decelerates to reach the terminal state with zero velocity, as shown in Figure~\ref{fig:velocityAcceleration}.



\subsection{Nonlinear Double Integrator}
In this section, we test the relaxed LMPC~\eqref{eq:RelaxedFTOCP} and \eqref{eq:policyLMPCRelaxed} on the following nonlinear double integrator problem
\begin{equation}\label{eq:NNDoubleIntegreator}
    \begin{aligned}
        \min_{T, a_0, \ldots, a_{T-1}} & ~ \sum_{t = 0}^{T-1} 1 \\
    \text{s.t. } \quad ~&  ~ \begin{bmatrix}x_{t+1} \\ v_{t+1}\end{bmatrix}  = \begin{bmatrix}x_{t} + v_t dt \\ v_{t} + \big(1-\frac{v_t^2}{ v_{\mathrm{max}}^2}\big)a_t dt \end{bmatrix}, \forall t \geq 0 \\
        & ~ 0 \leq v_t \leq v_{\mathrm{max}}, \forall t \geq 0 \\
        & ~ -1 \leq a_t \leq 1, \forall t \geq 0 \\
        & ~ x_T = x_F =[0, 0]^T,\\
        & ~ x_0 = x_S = [-10,0]^T,
    \end{aligned}
\end{equation}
where the state of the system are the velocity $v_t$ and the position $x_t$. The control action is the acceleration $a_t$
which is scaled by the concave function $g(v_t) = \big(1-{v_t^2}/{ v_{\mathrm{max}}^2}\big)$. 
In Section~\ref{sec:app:NN} of the Appendix we show that the above nonlinear double integrator satisfies Assumption~\ref{ass:simplified}.
We used a brute force algorithm to perform the first feasible trajectory used to initialize the relaxed LMPC policies synthesized with $N=4$. Furthermore, we implemented the strategy from Section~\ref{sec:dataReduction} using $P=\{ 12, 25, 50, 200\}$ data points and $i = \{1,3,4,10 \}$ iterations. 

Figures~\ref{fig:NNcompTime} shows the number of iterations needed to reach convergence. 
We notice that as more data points $P$ are used in the policy synthesis process, the closed-loop system convergence faster in the iteration domain to a trajectory which performs the task in $14$ time steps.


\begin{figure}[h!]
	\centering \includegraphics[trim = 0mm 0mm 15mm 14mm, clip, width=0.5\columnwidth]{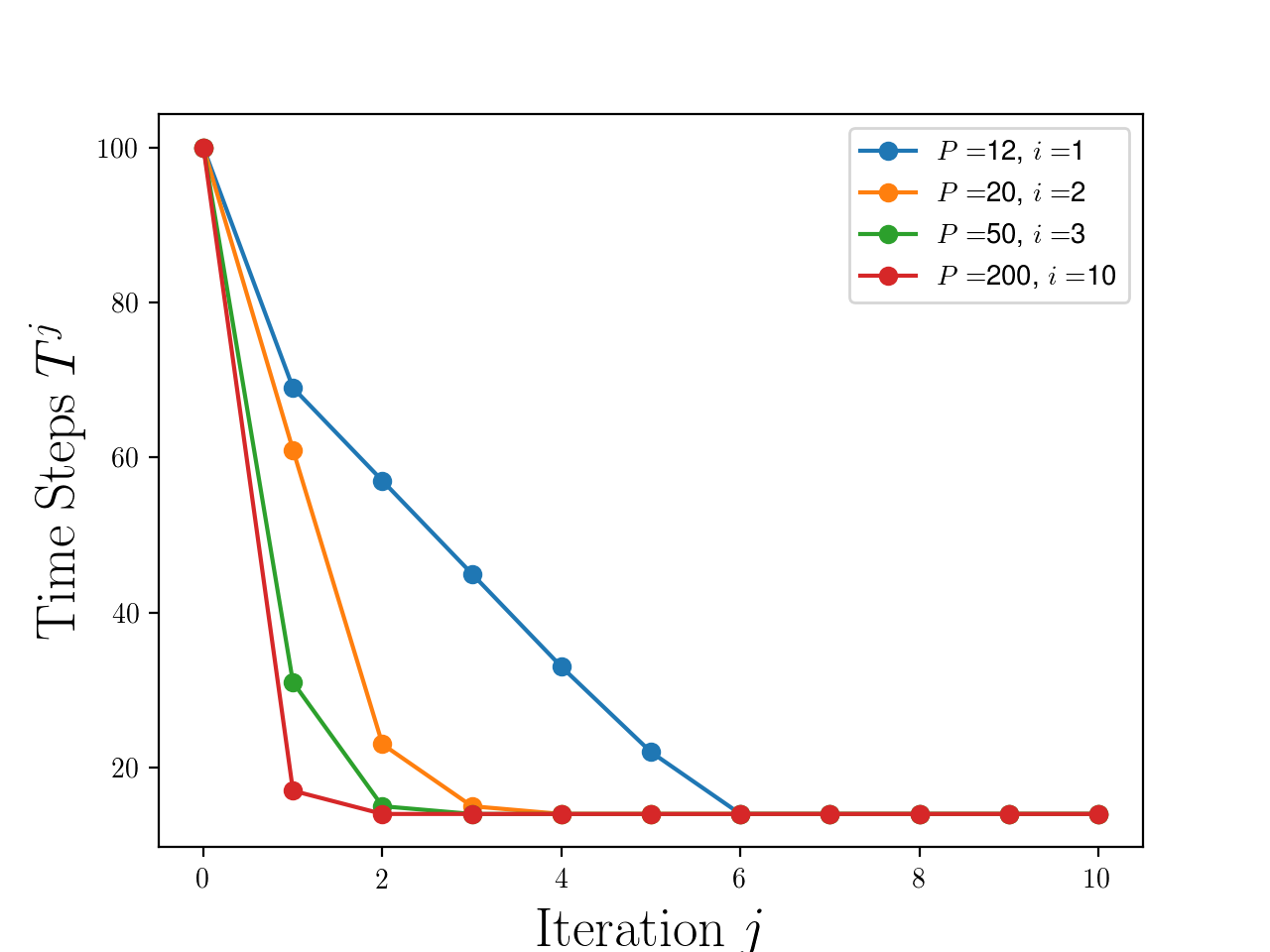}\caption{Time steps $T^j$ to reach $x_F$ as a function of the iteration index. We notice that, also in this example, as more data points are used in the synthesis process, the number of iterations needed to reach a steady state behavior decreases.}\label{fig:NNcompTime}
\end{figure}

Finally, Figures~\ref{fig:NN_cl} and \ref{fig:NN_input} show the steady-state closed-loop trajectories and the associated input sequences for all tested policies. We notice that after few iterations of the control task, all closed-loop systems converged to a similar behavior. In particular, the controller saturates the acceleration and deceleration constraints, as we would expect from the optimal solution to a time optimal control problem (Fig.~\ref{fig:NN_input}). 
It is interesting to notice that slowing down the nonlinear double integrator to zero speed requires more control effort than speeding up the system.
Therefore, the controller accelerates for the first $6$ time steps and then it decelerates for the last $8$ time steps to reach the terminal state with zero velocity.

\begin{figure}[h!]
	\centering \includegraphics[trim = 0mm 0mm 15mm 14mm, clip, width=0.5\columnwidth]{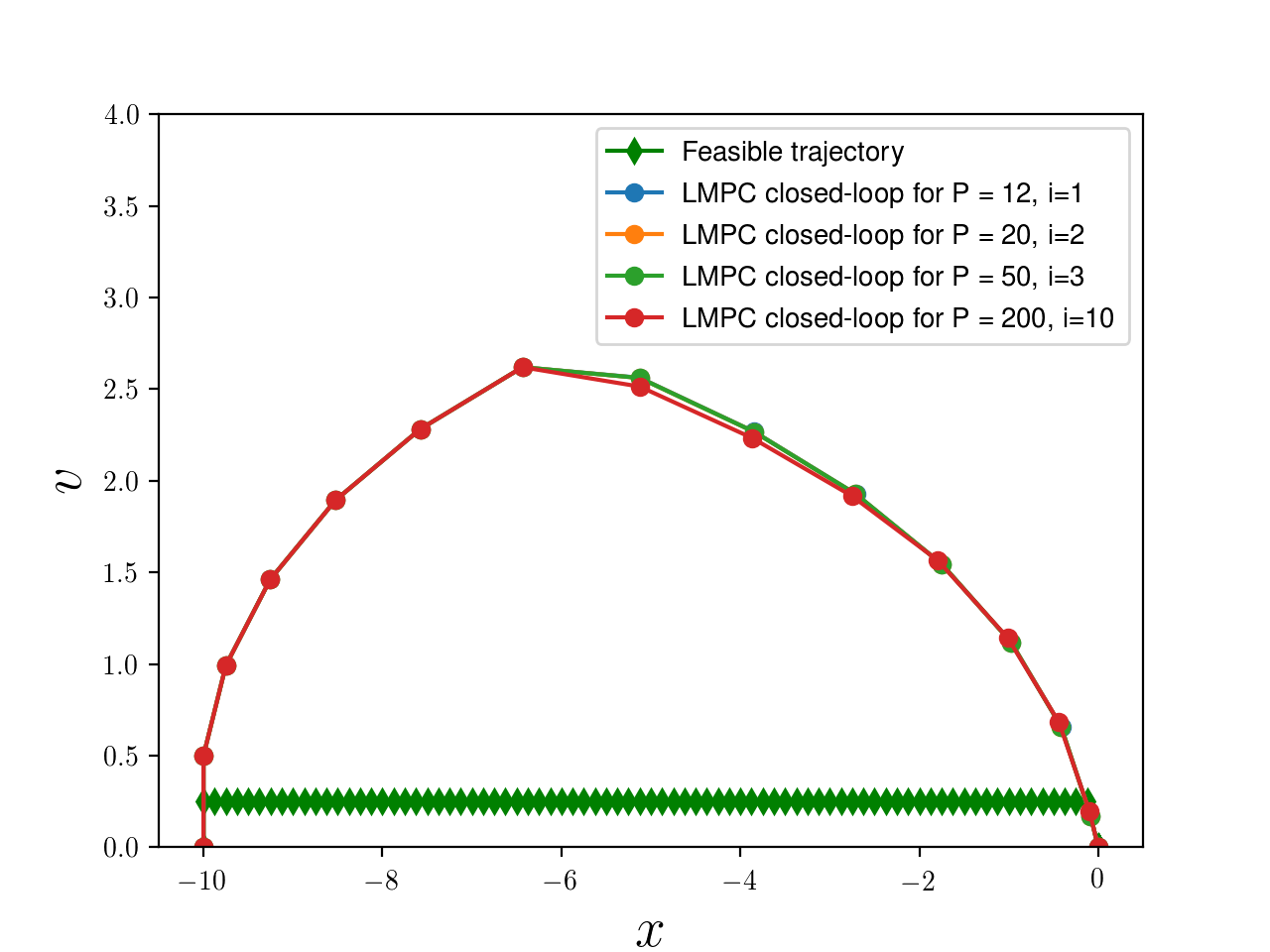}\caption{First feasible trajectory and closed-loop trajectories at the $10$th iteration. We notice that all LMPC policies converged to as similar behavior.}\label{fig:NN_cl}
\end{figure}
~\\
~\\
~\\
~\\
~\\
~\\
~\\
~\\

\begin{figure}[h!]
	\centering \includegraphics[trim = 0mm 0mm 15mm 14mm, clip, width=0.5\columnwidth]{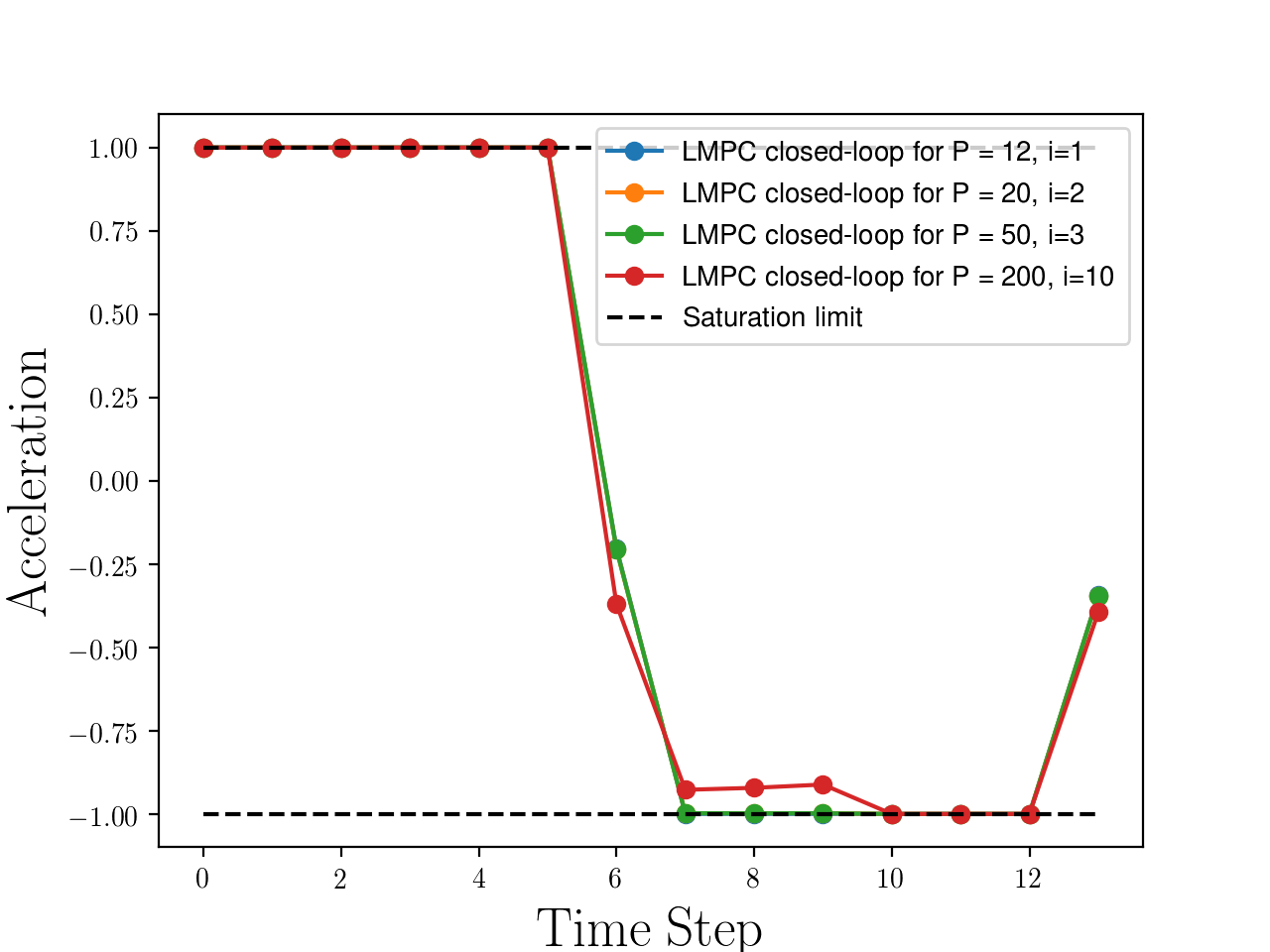}\caption{Acceleration inputs associated with the closed-loop trajectories at the $10$th iteration. We notice that the controller saturates the acceleration constraints.}\label{fig:NN_input}
\end{figure}

\newpage
\subsection{Minimum Time Dubins Car Racing}\label{sec:dubinsRacing}
We test the relaxed LMPC~\eqref{eq:RelaxedFTOCP} and \eqref{eq:policyLMPCRelaxed} on a minimum time racing problem.
The goal of the controller is to drive the dubins car on a curve of constant radius $R=10$ from the starting point $x_S$ to the finish line. More formally, our goal is to solve the following minimum time optimal control problem
\begin{equation}\label{eq:dubinsRacing}
    \begin{aligned}
        \min_{\substack{T, a_0, \ldots, a_{T-1} \\  \theta_0, \ldots, \theta_{T-1}}} & ~ \sum_{t = 0}^{T-1} 1 \\
    \text{s.t. } \quad ~ &  \begin{bmatrix}s_{t+1} \\ e_{t+1}\\ v_{t+1}\end{bmatrix}  = \begin{bmatrix}s_{t} + \frac{v_t \cos(\theta_t - \gamma(s_t))}{1 - e_t/R}dt \\ e_{t} + v_t\sin(\theta_t - \gamma(s_t))dt \\ v_{t} + a_t dt \end{bmatrix}, \forall t \geq 0 \\
        & ~ \begin{bmatrix} -2 \\ -1\end{bmatrix}\leq \begin{bmatrix} \theta_t \\ a_t\end{bmatrix} \leq \begin{bmatrix} 2 \\ 1\end{bmatrix}, \forall t\geq 0 \\
        & ~ e_{\mathrm{min}} \leq e_t \leq e_{\mathrm{max}}, \forall t \geq 0 \\
        & ~ x_T \in \mathcal{X}_F,\\
        & ~ x_0 = x_S = [0,0,0]^T,
    \end{aligned}
\end{equation}
where the states $s_t,e_t$ and $v_t$ are the distance travelled along the centerline, the lateral distance from the center of the lane and the velocity, respectively. Furthermore, $\gamma(s_t)$ is the angle of the tangent vector to the centerline of the road at the curvilinear abscissa $s_t$, the discretization time $dt = 0.5$s and the half lane width $e_{\mathrm{max}} = -e_{\mathrm{min}} = 2.0$. The control actions are the heading angle $\theta_t$ and the acceleration command $a_t$. Notice that the lane boundaries are represented by convex constraints on the state $e_t$, and therefore Assumption~\eqref{ass:convexity} is satisfied. The finish line is described by the following terminal set
\begin{equation}\label{eq:terminalSet}
    \mathcal{X}_F = \Bigg\{ x \in \mathbb{R}^3 \Bigg| \begin{bmatrix} 18.19 \\ -e_{\mathrm{min}} \\ 0 \end{bmatrix} \leq x \leq \begin{bmatrix} 18.69 \\ e_{\mathrm{min}} \\ 0 \end{bmatrix}\Bigg\}.
\end{equation}
As mentioned in Remark~\ref{remark:terminalSet}, in order to steer the system to a terminal set, we replaced $x_{T^i}^i = x_F$ with the vertices of $\mathcal{X}_F$ in definitions~\eqref{eq:CS} and \eqref{eq:cvxQfunction}.

In order to compute the first feasible trajectory needed to initialize the LMPC, we set $\theta_t^0= \gamma(s_t^0)$ and we designed a simple controller to steer the dubins car from $x_S$ to the terminal set $\mathcal{X}_F$. Notice that for $\theta_t^0= \gamma(s_t^0)$ the system is linear and consequently Assumption~\ref{property:storedDataAndDynamics} is satisfied for iteration $j=0$. For $j>0$, it is hard to verify analytically if Assumption~\ref{property:storedDataAndDynamics} holds, therefore we used a sampling strategy to approximately check this condition, as shown in the Appendix~\ref{sec:app:dubins}.

\begin{figure}[h!]
	\centering \includegraphics[trim = 0mm 0mm 15mm 14mm, clip, width=0.5\columnwidth]{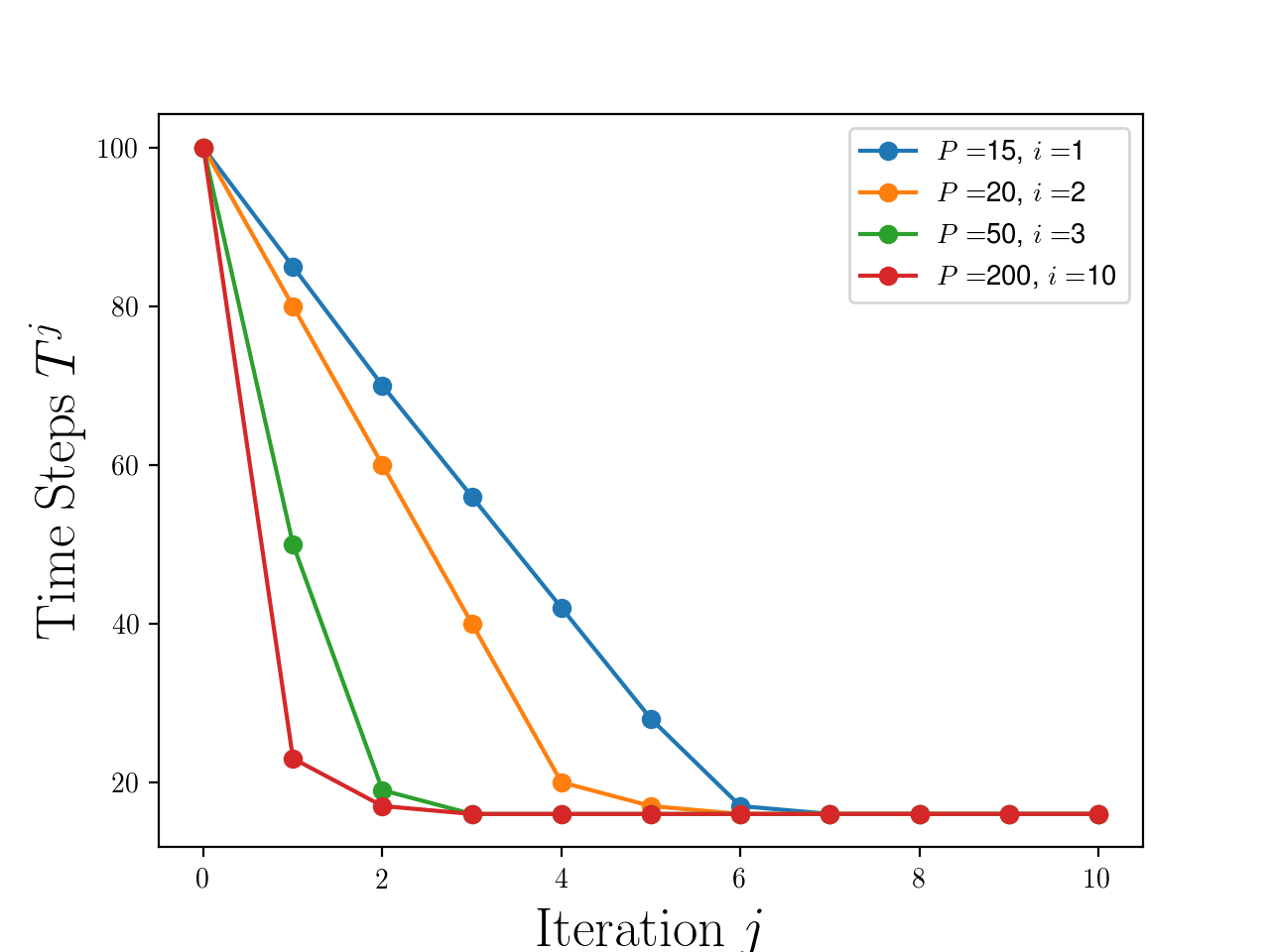}
	\caption{Time steps $T^j$ to reach $x_F$ as a function of the iteration index. We notice that as more points $P$ and iterations $i$ are used to synthesize the relaxed LMPC policy, the closed-loop system converges faster to a steady state behavior.}\label{fig:CSconvergenceSpeed}
\end{figure}

We test the LMPC policies synthesized with $N=4$ and using the strategy described in Section~\ref{sec:dataReduction} for $P=\{ 15, 25, 50, 200\}$ data points and $i = \{1,3,4,10 \}$ iterations. Figure~\ref{fig:CSconvergenceSpeed} shows the time steps $T^j$ needed to reach the terminal set~\eqref{eq:terminalSet}. We notice that after few iterations all LMPC policies converged to a steady state behavior which steers the system to the goal set in $16$ time steps. Also in this example, convergence is reached faster as more data points are used in the LMPC synthesis process. 



Furthermore, Figures~\ref{fig:CSclosedLoop} and \ref{fig:CSinputComparison} show that closed-loop trajectories and associated input sequences at convergence. In order to minimize the travel time, the LMPC cuts the curve and steers the system to a state within the terminal set which is close to the road boundary. Furthermore, we notice that the controller saturates the acceleration and deceleration constraints, as we expect from an optimal solution to a minimum time optimal control problem.

\begin{figure}[h!]
	\centering \includegraphics[trim = 0mm 0mm 15mm 14mm, clip, width=0.5\columnwidth]{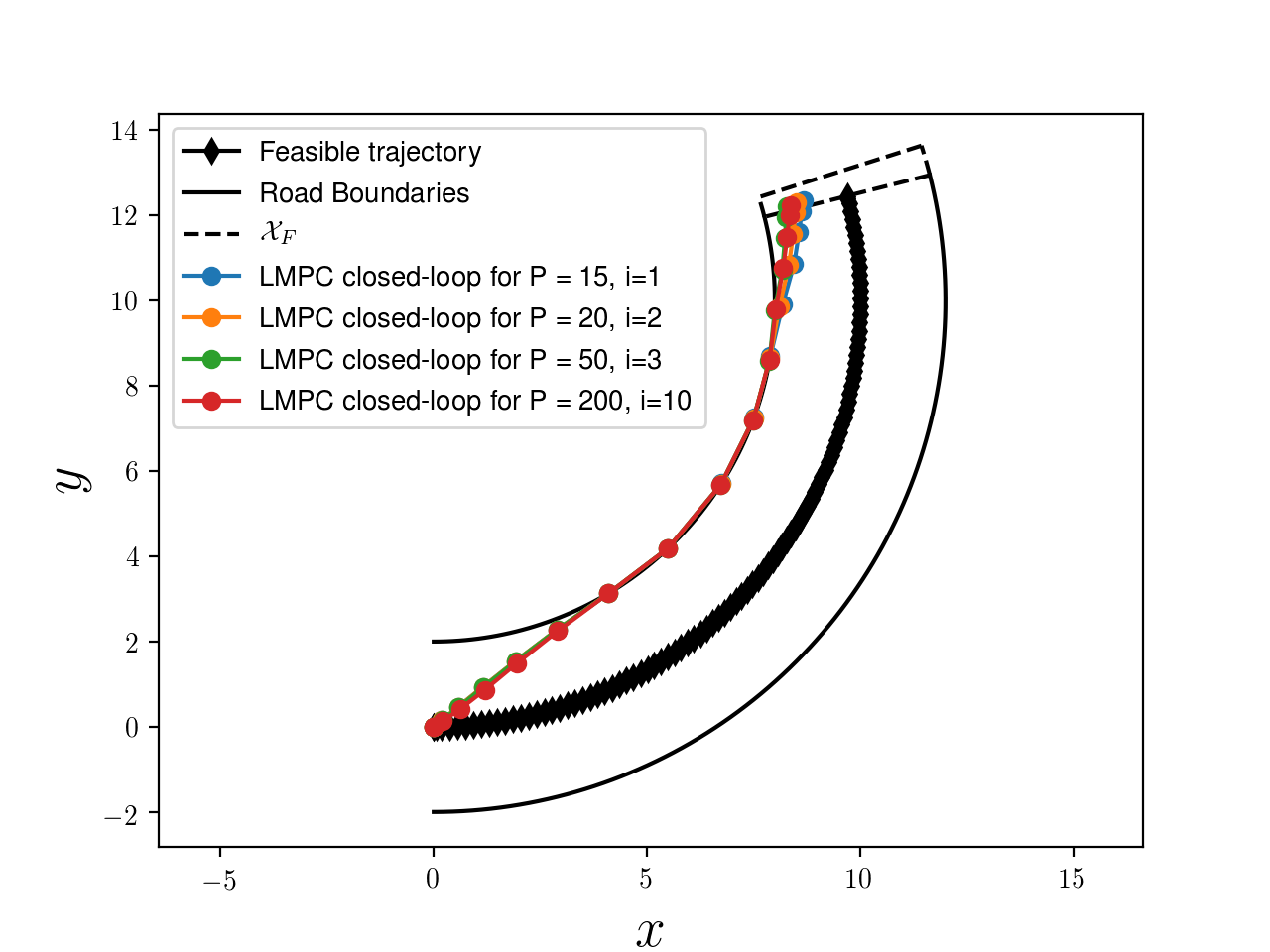}
	\caption{Comparison between the first feasible trajectory used to initialize the LMPC and the steady state LMPC closed-loop trajectories at convergence.}\label{fig:CSclosedLoop}
\end{figure}

\begin{figure}[h!]
	\centering \includegraphics[trim = 0mm 0mm 15mm 14mm, clip, width=0.5\columnwidth]{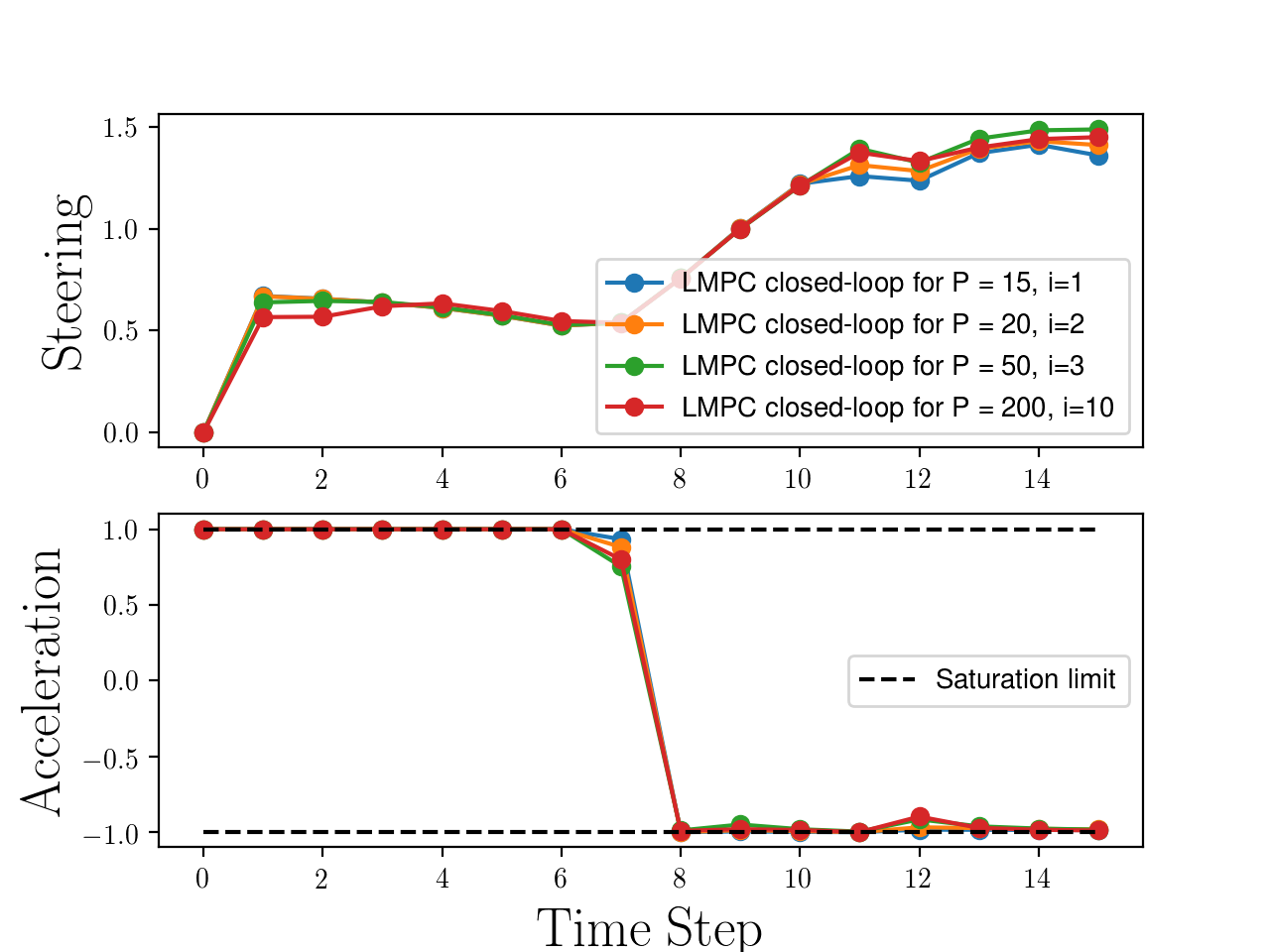}
	\caption{Comparison of the steady state inputs associated with the relaxed LMPC policies. We notice that the acceleration and deceleration is saturated, as we expect from the optimal solution to a minimum time optimal control problem.}\label{fig:CSinputComparison}
\end{figure}

Finally, we tested the LMPC policy starting from different initial conditions. In particular, we used the trajectories from Figure~\ref{fig:CSclosedLoop} to initialize the controller and we run $10$ iterations from the initial conditions reported in Table~\ref{tab:IC}. These initial conditions are contained in the $N$-steps controllable set from the set $\mathcal{CS}^0_{N}$. Therefore, as discussed in Theorem~\ref{th:diff_IC}, the LMPC policy is able to steer the system to the goal set while satisfying state and input constraints, as shown Figure~\ref{fig:IC}. Finally, we underline that the controller steered the system from all initial conditions to the terminal set in $16$ time steps.

\begin{table}[h!]
\centering
\caption{Initial conditions.}\label{tab:IC}
\begin{tabular}{c|cccccccccc}
Iteration & $j=1$& $j=2$& $j=3$& $j=4$& $j=5$& $j=6$& $j=7$& $j=8$& $j=9$& $j=10$ \\\midrule
Initial Condition & $\begin{bmatrix} 0.5 \\ 0.5 \\ 0.0 \end{bmatrix}$ &
                    $\begin{bmatrix} 0.15 \\ -1.0 \\ 0.0 \end{bmatrix}$ & 
                    $\begin{bmatrix} 0.1 \\ 0.3 \\ 0.0 \end{bmatrix}$ & 
                    $\begin{bmatrix} 0.0 \\ 0.0 \\ 0.0\end{bmatrix}$ & 
                    $\begin{bmatrix} 0.25 \\ 0.25 \\ 0.0 \end{bmatrix}$ & 
                    $\begin{bmatrix} 0.25 \\-0.25 \\ 0.0 \end{bmatrix}$ & 
                    $\begin{bmatrix} 0.5 \\ 0.0 \\ 0.0 \end{bmatrix}$ & 
                    $\begin{bmatrix} 0.0 \\ 0.25 \\ 0.0 \end{bmatrix}$ & 
                    $\begin{bmatrix} 0.25 \\ 0.0 \\ 0.0 \end{bmatrix}$ & 
                    $\begin{bmatrix} 0.15 \\ 0.2 \\ 0.0 \end{bmatrix}$         \\\midrule
\end{tabular}
\end{table}

\begin{figure}[h!]
	\centering \includegraphics[trim = 0mm 0mm 15mm 14mm, clip, width=0.5\columnwidth]{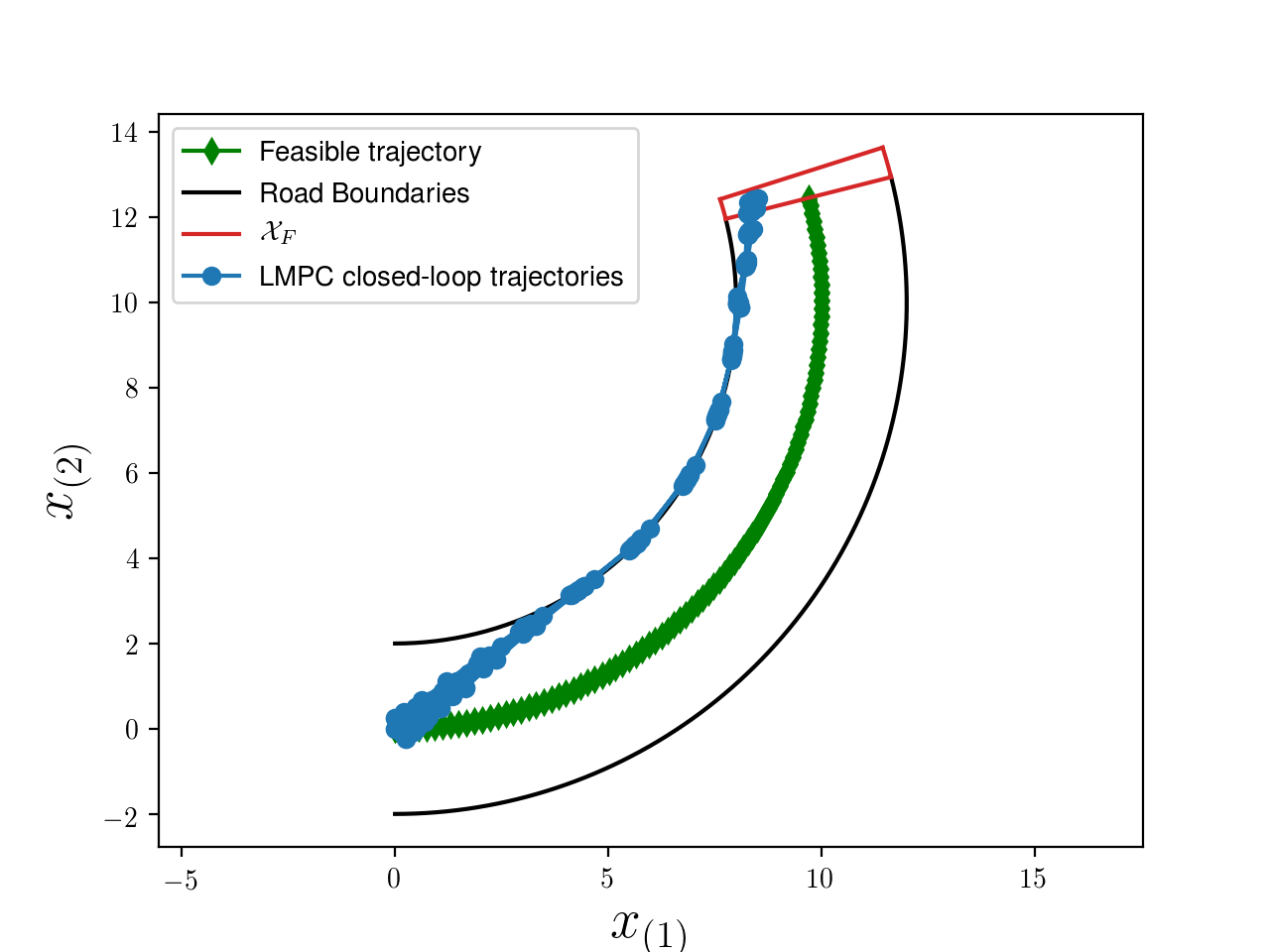}
	\caption{Closed-loop trajectories associated with the initial conditions from Table~\ref{tab:IC}.}\label{fig:IC}
\end{figure}


\section{Conclusions}\label{sec:conclusions}
We presented a time varying Learning Model Predictive Controller (LMPC) for time optimal control problems. The proposed control framework uses closed-loop data to construct time varying safe sets and approximations to the value function. Furthermore, we showed that these quantities can be convexified to synthesize a relaxed LMPC policy. We showed that the proposed control strategies guarantee safety, finite time convergence and non-decreasing performance with respect to previous task executions. Finally, we tested the controllers on three nonlinear minimum time optimal control problems.

\section{Acknowledgment}
The authors would like to thank Nicola Scianca from Sapienza University of Roma for the interesting discussions on repetitive LMPC and reviewers for helpful suggestions.
Some of the research described in this review was funded by the Hyundai Center of Excellence
at the University of California, Berkeley. This work was also sponsored by the Office of Naval
Research gran N00014-18-1-2833. The views and conclusions contained herein are those of the authors and should not be
interpreted as necessarily representing the official policies or endorsements, either expressed or
implied, of the Office of Naval Research or the US government.

\section{Appendix}

\subsection{Nonlinear Double Integrator}\label{sec:app:NN}
In this section, we show that the following nonlinear double integrator
\begin{equation*}
    z_{k+1} = \begin{bmatrix}x_{k+1} \\ v_{k+1} \end{bmatrix} = \begin{bmatrix}x_{k} +  v_k dt\\ v_{k} + g(v_k) a_k dt\end{bmatrix} = f_n(z_k, a_k) 
\end{equation*}
for $g(v_k) = (1 - v_k^2/v_{\mathrm{max}}^2)$ satisfies Assumption~\ref{ass:simplified}. Consider a set of states $x^{(i)} \in \mathcal{X}$, inputs $u^{(i)} \in \mathcal{U}$ and multipliers $\lambda^{(i)} \geq 0$, for $i\in \{1,\ldots, P\}$. Let
\begin{equation*}
    x = \sum_{k=1}^P \lambda^{(k)} x^{(k)} \text{ and } \sum_{k=1}^P \lambda^{(k)} =1,
\end{equation*}
we have that
\begin{equation*}
    \sum_{k=1}^P \lambda^{(k)} f_n(z^{(k)}, a^{(k)}) = \sum_{k=1}^P  f_n(\lambda^{(k)} z^{(k)}, a)
\end{equation*}
where 
\begin{equation*}
    a = \frac{\sum_{k=1}^P \lambda^{(k)} g(v^{(k)}) a^{(k)}}{ g\big( \sum_{k=1}^P \lambda^{(k)} v^{(k)} \big) }.
\end{equation*}
Finally, by concavity of $g(v_k)\geq 0 $ for all $z_k = [x_k, v_k]^T \in \mathcal{X}$ we have that
\begin{equation*}
    a = \frac{\sum_{k=1}^P \lambda^{(k)} g(v^{(k)}) a^{(k)}}{ g\big( \sum_{k=1}^P \lambda^{(k)} v^{(k)} \big) } \geq \frac{ \sum_{k=1}^P \lambda^{(k)} g(v^{(k)})  }{g\big( \sum_{k=0}^P \lambda^{(k)} v^{(k)}\big) } a_{\textrm{min}} \geq a_{\textrm{min}} \text{ and } a = \frac{\sum_{k=1}^P \lambda_k g(v^{(k)}) a^{(k)}}{ g\big( \sum_{k=1}^P \lambda^{(k)} v^{(k)} \big) } \leq \frac{ \sum_{k=1}^P \lambda^{(k)} g(v^{(k)})  }{g\big( \sum_{k=0}^P \lambda^{(k)} v^{(k)}\big) } a_{\textrm{max}} \leq a_{\textrm{max}}
\end{equation*}
where $a_{\textrm{min}}=-1$ and $a_{\textrm{max}}=1$. Therefore, we conclude that $a \in \mathcal{U}$ and Assumption~\ref{ass:simplified} is satisfied.

\begin{figure}[t]
\centering
\begin{minipage}{.5\textwidth}
  \centering \includegraphics[trim = 36mm 12mm 14mm 16mm, clip, width=0.95\columnwidth]{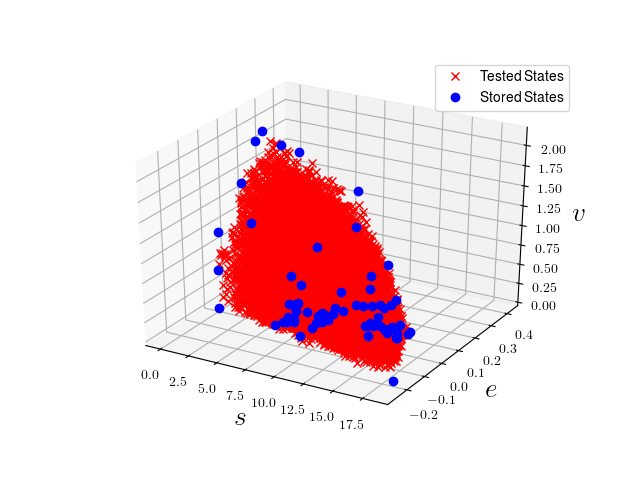}
	\caption{Randomly sampled states used to check \\ that Assumption~\ref{property:storedDataAndDynamics} is approximately satisfied.}\label{fig:testIt3}
\end{minipage}%
\begin{minipage}{.5\textwidth}
	\centering \includegraphics[trim = 36mm 12mm 14mm 16mm, clip, width=0.95\columnwidth]{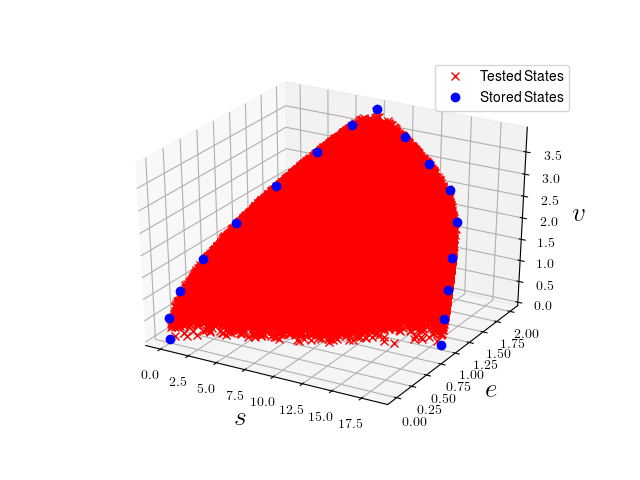}
	\caption{Randomly sampled states used to check that Assumption~\ref{property:storedDataAndDynamics} is approximately satisfied.}\label{fig:testIt9}
\end{minipage}
\end{figure}

\subsection{Dubins Car}\label{sec:app:dubins}
We used a sampling strategy to check if Assumption~\ref{property:storedDataAndDynamics} is approximately satisfied for the example in Section~\ref{sec:dubinsRacing}. In particular, for $s \in \{1,\ldots,10^5\}$ we randomly sampled a set of states $[ x^{(1,s)},\ldots,   x^{(n+1,s)}]$ from the set of stored states $\big\{ \bigcup_{i=0}^j \bigcup_{t=0}^{T^i} x_t^i \big\}$ and a set of multipliers $[ \lambda^{(1,s)},\ldots,  \lambda^{(n+1,s)}] \in \Lambda$ where $\Lambda = \{ [ \lambda^{(1)},\ldots, \lambda^{(n+1)}]: \lambda^{(k)}\geq 0, \sum_{k=1}^{n+1}=1\}$. Afterwards, we checked that for all $s \in \{1,\ldots,10^5\}$ Assumption~\ref{property:storedDataAndDynamics} is satisfied at the sampled points, i.e., $\forall s \in \{1,\ldots,10^5\}$
\begin{equation}\label{eq:sampledDataCheck}
    \exists u \in \mathcal{U} \text{ such that } f\big(x^{(s)},u\big) \in \text{Conv}\Big( \textstyle \bigcup_{s=1}^{n+1} f( x^{(k,s)}, u^{(k,s)})\Big),
\end{equation}
where $x^{(s)} = \sum_{k=1}^{n+1} \lambda^{(k,s)} x^{(k,s)}$ and $ u^{(k,s)}$ is the stored input associated with the stored state $x^{(k,s)}$. As~\eqref{eq:sampledDataCheck} is satisfied for all $10^5$ randomly sampled data points and $10^5 \geq \log(1/\beta) / \log(1/\epsilon)$ for $\beta = 10^{-6}$ and $\epsilon = 0.99986$. From \cite[Proposition~1]{zhang2019safe} and \cite[Theorem~3.1]{tempo1996probabilistic}, we have that with confidence $\beta = 10^{-6}$ the probability of randomly sampling a set of states $[ x^{(1)},\ldots,  x^{(n+1)}]$ and a set of multipliers $[ \lambda^{(1,i)},\ldots,  \lambda^{(n+1,i)}] \in \Lambda$ for which Assumption~\ref{property:storedDataAndDynamics} is satisfied is $\epsilon = 0.99986$, i.e.,
\begin{equation*}
\begin{aligned}
    \mathbb{P}\Bigg[&\exists u \in \mathcal{U} \text{ such that }f\big(\textstyle \sum_{k=1}^{n+1} \lambda^{(k)} x^{(k)},u \big) \in \text{Conv}\Big( \textstyle \bigcup_{k=1}^{n+1} f( x^{(k)}, u^{(k)})\Big)\Bigg]\geq \epsilon=0.99986
\end{aligned}
\end{equation*}
where the random variable $ x^{(k)}$ has support $\big\{ \bigcup_{i=0}^j \bigcup_{t=0}^{T^i} x_t^i \big\}$ and the vector of random variables $[ \lambda^{(1)},\ldots, \lambda^{(n+1)}]$ has support $\Lambda = \{ [ \lambda^{(1)},\ldots, \lambda^{(n+1)}]: \lambda^{(k)}\geq 0, \sum_{k=1}^{n+1}=1\}$. Both $ x^{(k)}$ and $[ \lambda^{(1)},\ldots, \lambda^{(n+1)}]$ have the same distributions that were used to generate the data points in~\eqref{eq:sampledDataCheck}.  Finally, for $j=\{3, 10\}$ Figures~\ref{fig:testIt3} and~\ref{fig:testIt9} show the $10^{5}$ randomly generated states $x^{(s)} = \sum_{k=1}^{n+1}\lambda^{(k,s)} x^{(k,s)}$ where we have verified that Assumption~\ref{property:storedDataAndDynamics} holds.







\bibliography{mybibfile}%




\end{document}